\documentclass{amsart}
\usepackage{amsmath,amssymb,graphicx,epsfig}
\usepackage{amsmath}
\usepackage{amssymb}
\usepackage{fancybox}
\usepackage{color}

 \newtheorem{thm}{Theorem}[section]
 \newtheorem{cor}[thm]{Corollary}
 \newtheorem{lem}[thm]{Lemma}
 \newtheorem{prop}[thm]{Proposition}
 \newtheorem{claim}[thm]{Claim}
 \newtheorem{defn}[thm]{Definition}
 \newtheorem{rem}[thm]{Remark}

 \newcommand{\set}[1]{\left\{#1\right\}}

\newcommand{\eps}{\varepsilon}

\newcommand \be     {\begin{equation}}
\newcommand \ee     {\end{equation}}
\newcommand {\RR} {\mathbb{R}}
\newcommand {\CC} {\mathbb{C}}
\newcommand {\RT }   {\mathbb{R}^3}
\newcommand {\Rn}   {\mathbb{R}^n}
\newcommand \Hcal   {\mathcal H}

\newcommand \Lcal   {\mathcal L}

\newcommand \Hcalth   {{\mathcal H}^\theta}
\newcommand{\MX}{\mathcal{X}}

\newcommand \Zbar {\overline{Z}}

 \newcommand{\half}{\frac{1}{2}}

 \newcommand{\pa}{\partial}
 \newcommand{\sbar}{\overline{s}}
 \newcommand{\paxj}{\partial x_j}
 \newcommand{\suml}{\sum\limits}
 \newtheorem*{THEOREM S}{REGULARITY ASSUMPTION ON SPECTRAL DENSITY}
 \newtheorem{conjecture}[thm]{CONJECTURE}
 \numberwithin{equation}{section}
\begin{document}

\title [FIRST-ORDER SYSTEMS -DIRAC TO CRYSTALS]
 {  SPECTRAL THEORY OF FIRST-ORDER SYSTEMS: FROM CRYSTALS TO DIRAC OPERATORS}

\author{Matania Ben-Artzi}
\address{Matania Ben-Artzi: Institute of Mathematics, The Hebrew University, Jerusalem 91904, Israel}
\email{mbartzi@math.huji.ac.il}
\author{Tomio Umeda}
\address{Tomio Umeda: Department of
Mathematical Sciences,  University of Hyogo, Himeji 671-2201,
Japan}
\email{	umeda@sci.u-hyogo.ac.jp}


\thanks{This work was started at the Newton Institute, Cambridge UK, when the two authors were participants in the programme ``Spectral Theory of Relativistic Operators'' (Summer 2012), organized by M. Brown, M. J. Esteban, K. Schmidt and H. Siedentop. The authors thank the Newton Institute and the organizers for their support.  The University of Hyogo (Himeji) and Ritsumeikan University supported M. Ben-Artzi in his visit (March-April 2017) and the Hebrew University supported T. Umeda in his visit (January 2018), when this paper was completed. Special thanks are due to Professor H. Isozaki of Ritsumeikan University for his support and Professor H. Tamura for his valuable suggestions and constant interest in this work. We thank Professor V. Mazya for bringing the reference ~\cite{hormander-lions} to our attention.
T.Umeda was partially supported by
 the Japan Society for the Promotion of Science
     ``Grant-in-Aid for Scientific Research'' (C)
    No.  26400175.
    \color{black}}

\subjclass[2010]{Primary 35Q41; Secondary 35Q61, 35Q40}

\keywords{Dirac operator, Maxwell equations, first-order systems, strongly propagative systems, spectral derivative, spacetime estimates, spectrum, limiting absorption principle, thresholds, eigenvalues in gap, perturbation  by potential}

\date{\today}






\begin{abstract}
   Let
   $$L_0=\suml_{j=1}^nM_j^0D_j+M_0^0,\,\,\,\,D_j=\frac{1}{i}\frac{\pa}{\paxj},
       \quad x\in\Rn,$$
       be a constant coefficient
   first-order partial differential system, where the matrices $M_j^0$ are Hermitian. It is assumed that the homogeneous part is strongly propagative. In the nonhomegeneous case it is assumed that the operator is isotropic  . The spectral theory of such systems and their potential perturbations is expounded, and a Limiting Absorption Principle  is obtained up to thresholds. Special attention is given to a detailed study of the Dirac and Maxwell operators.

   The estimates of the spectral derivative near the thresholds are based on detailed trace estimates on the slowness surfaces. Two applications of these estimates are presented:
   \begin{itemize}
   \item Global spacetime estimates of the associated evolution unitary groups, that are also commonly viewed as decay esimates. In particular  the Dirac and Maxwell systems are explicitly treated.
       \item The finiteness of the eigenvalues (in the spectral gap) of the perturbed Dirac operator is studied, under suitable decay assumptions on the potential perturbation.
\end{itemize}
\end{abstract}
\maketitle
\section{\textbf{I
INTRODUCTION}}
   The equations of classical physics, governing acoustic, electromagnetic and elastic waves
    in  anisotropic media, are given as first-order hyperbolic systems. Similarly, the Dirac equation
    of relativistic quantum electrodynamics is such a system. Already in the classical treatise of Courant and Hilbert
    the common features of  these systems are brought to the fore  ~\cite[Chapter III, \S2, Chapter VI, \S3a]{courant}.

     Motivated by this approach,
    emphasizing a unifying point of view, we study here basic spectral properties of a
    class of first-order self-adjoint operators, that includes the aforementioned systems,
    or, more explicitly, their spatial generators.

    Thus, we consider operators acting on $\CC^K$-valued functions in $\Rn,\,n\geq 2,$  having the form
    \begin{equation}\label{systemL}L=\suml_{j=1}^nM_j^0D_j+M^0_0+V(x),\,\,D_j=\frac{1}{i}\frac{\pa}{\paxj},
    \quad x\in\Rn,
\end{equation}
      where  the matrices $M_j^0,\,0\leq j\leq n,$ are Hermitian $K\times K$  constant matrices
      and the matrix $V(x)$
      is Hermitian for every $x\in\Rn.$  Thus $L$ is symmetric (with respect to the $\Lcal^2(\Rn,\CC^K)$ scalar product)
       on $C^\infty_0(\Rn,\CC^K).$

      The obvious notation $$\Lcal^2(\Rn,\CC^K)=\Lcal^2(\Rn)\bigotimes
      \mathbb{C}^K$$
      (resp. $C^\infty_0(\Rn,\CC^K)=C^\infty_0(\Rn)\bigotimes \mathbb{C}^K$) has been used
      for the space of Lebesgue square-integrable (resp. smooth and compactly supported) $\CC^K$-valued functions.

      Our objectives in this paper are:
    \begin{itemize}

      \item Establish a ``Limiting Absorption Principle'' (LAP) for  the special cases of the Dirac and Maxwell operators (Section ~\ref{sec-spec-dirac-maxwell}) up to the thresholds, with sharp resolvent estimates.

           \item Generalize the LAP to ``strongly propagative'' or  ``isotropic operators'' (Section ~\ref{secspecunperturbed}).

       \item  Extend the LAP to the case of potential perturbation (Section ~\ref{perturbed-sec}).

    \item Exploit the sharp   resolvent estimates of the Dirac and Maxwell operators near the  thresholds in the study of two applications:

         \begin{itemize}
          \item Obtain decay conditions on the potential perturbation of the Dirac operator which guarantees the finiteness of the isolated eigenvalues in the spectral gap (Subsubsection ~\ref{secfiniteevs}).
          \item Derive  global spacetime estimates of the solutions in weighted $L^2$ spaces (Section ~\ref{secspacetime}).
              \end{itemize}
   \item Extend the  global spacetime estimates of the solutions in weighted $L^2$ spaces to the case of strongly propagative systems (Section ~\ref{secspacetimestrongprop}).

    \end{itemize}

          Both the Dirac (Subsection ~\ref{subsec-dirac}) and the Maxwell (Subsection ~\ref{subsec-maxwell}) systems are isotropic. Their prime significance in physics justifies the detailed study of their spectral structure, as is carried out in Section ~\ref{sec-spec-dirac-maxwell}.

           The operator $L$ is viewed as a perturbation of a constant coefficient, symmetric operator
       \begin{equation}\label{systemL0}L_0=\suml_{j=1}^nM_j^0D_j+M_0^0,\,\,\,\,D_j=\frac{1}{i}\frac{\pa}{\paxj},
       \quad x\in\Rn.
\end{equation}
       The cases where $L_0$ is either \textit{strongly propagative} (Definition ~\ref{defn-strong-prop})  or \textit{isotropic}
        (Definition ~\ref{def-spherical-symbol}) are studied in Section ~\ref{secspecunperturbed}.

         It will be seen that the LAP
    is closely connected to the geometry of the level sets of the characteristic surfaces (``normal'',
    by the terminology of ~\cite{courant}) in $\Rn.$
     In the physical context they are  referred to as ``Slowness Surfaces''~\cite[Section 4]{wilcox1}.

          The presence of the non-zero matrix $M_0^0$  is motivated by the ``massive'' Dirac operator (Subsection ~\ref{subsec-dirac}). The basic operator $L_0$ is then nonhomogeneous, meaning that the associated eigenvalues are not homogeneous functions (see Section  ~\ref{secspecunperturbed} for details) and as a result the geometry of the ``slowness surfaces'' is more complex. As far as we know, the Dirac operator is the only instance of a nonhomogeneous operator that has been treated in the literature. In this paper we study the nonhomogeneous operator in the general isotropic case.

          A general (potential) perturbation theory
      is  presented in Section ~\ref{perturbed-sec}.
        As an immediate consequence of the LAP for the perturbed operator,
    it follows that the spectrum is absolutely continuous, except possibly for a discrete sequence of embedded eigenvalues.

        The general theory is then applied to the perturbed Dirac operator, as an important example.  The threshold estimates for the perturbed Dirac operator enable us (Subsubsection ~\ref{secfiniteevs})
to give criteria for the finiteness of eigenvalues in the spectral gap.

 Remark that the perturbed Maxwell system can also be reduced to the case of potential perturbation ~\cite[Section 1.4]{tamura1}, but we choose not to treat it here in detail, as the paper is already quite long. We refer to ~\cite{birman, filonov} and references therein for the study of the perturbed Maxwell operator in terms of self-adjointness and absolute continuity of the spectrum.

     In Section ~\ref{secspacetime} we establish global spacetime estimates for the Dirac and Maxwell operators, and in Section ~\ref{secspacetimestrongprop} such estimates are derived for general (homogeneous) strongly propagative operators.

    The literature concerning various aspects of the spectral and scattering theory of first-order systems, as well as asymptotic decay of their solutions, is very extensive, hence our reference list is far from being comprehensive. Certainly there are many papers that well deserve being included in our list and of which we are not aware.  However, we have made an effort to include references to works directly related to this paper and dealing primarily with the LAP (in interior intervals of the spectrum or at thresholds). Most of these papers were concerned with the Dirac operator. We refer to ~\cite{murata,rauch,smith,wilcox} and references therein for classical treatments of decay of solutions of first-order symmetric systems, and to ~\cite{mochizuki} in the case of exterior domains.

    The first proof of the LAP for the (massive) Dirac operator with short range potential perturbation was obtained by Yamada~\cite{yamada0}. It was subsequently proved (imposing various hypotheses on the perturbation) by several authors (see ~\cite{balslev,boussaid,boussaid1,PSU, yamada} for the massive case and ~\cite{ancona1,SU1} for the zero mass case), as well as the recent papers ~\cite{carey,erdogan} (and references therein) for both the massive and zero mass cases.  Their treatments rely on the fact (see Equation ~\eqref{squaredirac}) that the square of this operator is the Schr\"{o}dinger operator. A weak* form of the LAP, using the methodology of conjugate operators, was obtained in ~\cite{berthier2}. The LAP (for the Dirac operator) up to the threshold was proved in ~\cite{ifimovici}. Recently their method has been extended in ~\cite{carey} in order to obtain a global LAP for the massless Dirac Operator in all dimensions.  The zero modes and zero resonances of massless Dirac operators were studied in ~\cite{SU1}
    and eigenfunctions at the threshold energies were studied in ~\cite{SU2}.
     We refer also to ~\cite{ba11,umeda} for the closely related ``relativistic Schr\"{o}dinger operator''.

    For the Maxwell equations in crystal optics the LAP was established in ~\cite[Theorem 1.2]{tamura1}.

    For a uniformly propagative system (see Definition ~\ref{defn-uniform-prop}) of the form $E(x)^{-1}\suml_{j=1}^nM_j^0D_j$ the LAP was established in ~\cite{tamura,yajima}, and also in ~\cite{Schulenberger}.  For the more general strongly propagative systems it was proved in   ~\cite[Lemma 2.1]{weder}. See however Remark ~\ref{remLAPWeder} concerning this paper. These works dealt with interior intervals of the spectrum, whereas here we obtain detailed estimates at the thresholds.

    We are not aware of any work where the LAP for the nonhomogeneous  operator $L_0$  (namely, $M_0^0\neq 0$) has been addressed, other than the massive Dirac operator.

    We refer to Remark ~\ref{remfiniteevs} concerning previous studies of the finiteness of eigenvalues in the spectral gap, in the case of the perturbed (massive) Dirac operator.

\section{\textbf{THE CLASS OF UNPERTURBED OPERATORS--DIRAC AND MAXWELL SYSTEMS}}\label{unperturbed-sec}
  The coefficients $M_0^0,\,M_1^0,\ldots,M_n^0$ of the unperturbed  operator $L_0$ ~\eqref{systemL0}
  are constant Hermitian $K\times K$ matrices (over $\CC$). The addition of $M_0^0$ will enable us in particular to
  include the massive Dirac operator in our treatment.

  The (unitary) Fourier transform is defined by

        \begin{equation}\label{fourier}(\mathcal{F}u)(\xi)= \widehat{u}(\xi)=(2 \pi)^{-\frac{n}{2}} \int
         \limits_{\Rn} u(x) e^{-i <\xi, x>_{\Rn}} dx .\end{equation}

           The constant coefficient operators are transformed into multiplication operators (by \textit{symbols}).

  The homogeneous part of $L_0$ is assumed to be strongly propagative~\cite{wilcox},
  according to Definition ~\ref{defn-strong-prop} below.

        We shall address this general case in Section ~\ref{secspecunperturbed}. We start here with the two most famous physical examples of such operators, namely, the Dirac and Maxwell systems. As we shall see, both systems (including the nonhomogeneous Dirac system with mass) share the property of being \textit{isotropic} (Definition
        ~\ref{def-spherical-symbol} below).

\subsection{\textbf{EXAMPLE: THE FREE DIRAC OPERATOR}}\label{subsec-dirac} As a special case  we consider
the free Dirac operator. It is the self-adjoint operator $H_m$ in
$\Lcal^2(\RT,\mathbb{C}^4)=\Lcal^2(\RT)\bigotimes \mathbb{C}^4$
given by
\begin{equation}\label{diracop}H_m=\alpha\cdot D+m\beta,\,\,m\geq 0,
\end{equation}
where \be\label{del} D=\frac{1}{i} \nabla_x,\quad x\in \RT, \ee
and $\alpha=(\alpha_1,\alpha_2,\alpha_3)$ is the triplet of
$4\times 4$ Dirac matrices \be
\alpha_j=\begin{pmatrix}
  O_2 & \sigma_j \\
  \sigma_j& O_2
\end{pmatrix},\quad j=1,2,3.
\ee Here $O_2$ is the $2\times 2$ zero matrix and the $2\times 2$ matrices $\sigma_j$ (Pauli matrices) are given
by \be \sigma_1=\begin{pmatrix}
       0&1 \\
       1 & 0\\
     \end{pmatrix},\,\,\,\sigma_2=\begin{pmatrix}
       0&-i \\
       i & 0\\
     \end{pmatrix},\,\,\,\sigma_3=\begin{pmatrix}
       1&0 \\
       0 & -1\\
     \end{pmatrix},\ee
     and
     \be\beta=\begin{pmatrix}
       I_2&O_2 \\
       O_2& -I_2\\
     \end{pmatrix},\,\,I_2=\begin{pmatrix}
       1&0 \\
       0& 1\\
     \end{pmatrix}.
     \ee
      The symbol $M_m(\xi)$ (corresponding to $M_0(\xi)$ above) is a $4\times 4$
    Hermitian matrix given by
    \be
M_m(\xi)=\alpha\cdot\xi+m\beta.
    \ee
     It is readily verified that
     we have the following equality of  self-adjoint operators  in
$L^2(\RT,\mathbb{C}^4),$
\begin{equation}\label{squaredirac}(H_m)^2=(-\Delta+m^2)\bigotimes
I_4.
\end{equation}

    \begin{claim}
     The eigenvalues of $M_m(\xi)$ are given by
    $\lambda_\pm(\xi)=\pm\sqrt{|\xi|^2+m^2},$ and are both of
    double multiplicity (except for $m=0$ and $\xi=0$).
       \end{claim}

       It follows that the homogeneous operator $H_0$ is uniformly propagative (Definition ~\ref{defn-uniform-prop}).

     For any
    $\xi\in\RT$ (assuming either $\xi\neq 0$ or $m\neq 0$) there exists a unitary matrix $U_m(\xi)$ such
    that
    \be\label{eqdiraccoerciv}
      U_m(\xi)^*M_m(\xi)U_m(\xi)=\begin{pmatrix}
       \lambda_+(\xi)I_2& O_2\\
       O_2& \lambda_-(\xi)I_2\\
     \end{pmatrix}.
    \ee

  For  $f\in \Lcal^2(\RT,\mathbb{C}^4)$ we define the transformation
    \be\label{eqdefGm}
    (\mathcal{G}_mf)(\xi)=U_m^{\ast}(\xi)\widehat{f}(\xi),\,\,\xi\in\RT.
        \ee
        Then
$$\mathcal{G}_m:\Lcal^2_x(\RT,\mathbb{C}^4)\to \Lcal^2_{\xi}(\RT,\mathbb{C}^4)\quad,$$
   is unitary and diagonalizes $H_m$ in the sense that $\mathcal{G}_mH_m\mathcal{G}_m^{-1}$ in $\Lcal^2_{\xi}(\RT,\mathbb{C}^4)$ is given by the diagonal multiplication operator
  \be\label{eqdiagHm}
   \Big(\mathcal{G}_mH_m\mathcal{G}_m^{-1}\widehat{f}\Big)(\xi)=\begin{pmatrix}
       \lambda_+(\xi)I_2&O_2 \\
       O_2& \lambda_-(\xi)I_2\\
     \end{pmatrix}\widehat{f}(\xi).
  \ee

    In fact, in the physical literature this transformation is known as the \textbf{ Foldy-Wouthuysen-Tani transformation}
    ~\cite{foldy}.
      The transformation is explicitly presented (as can be easily verified) by ~\cite[Section 1.4]{thaller} and ~\cite[Section 2.1]{balinsky}:
     \be\label{eqexactUm}\aligned
     U_m(\xi)=\frac{1}{\sqrt{2}\sqrt{\lambda_+(\xi)^2+m\lambda_+(\xi)}}\Big\{(\lambda_+(\xi)+m)I_4+
     \beta(\alpha\cdot\xi)\Big\}\hspace{50pt}\\
     \\=\frac{1}{\sqrt{2}\sqrt{\lambda_+(\xi)^2+m\lambda_+(\xi)}}\left(
                                                                \begin{array}{cccc}
                                                                  \lambda_+(\xi)+m & 0 & -\xi_3 & -\xi_1+i\xi_2 \\
                                                                  0 & \lambda_+(\xi)+m & -\xi_1-i\xi_2 &\xi_3\\
                                                                  \xi_3 & \xi_1-i\xi_2 & \lambda_+(\xi)+m & 0 \\
                                                                  \xi_1+i\xi_2 & -\xi_3 & 0 & \lambda_+(\xi)+m \\
                                                                \end{array}
                                                              \right)
     \endaligned\ee

     Clearly, the columns of this matrix are the eigenvectors of the symbol matrix $M_m(\xi).$

     We shall need this transformation when studying the spectral structure of the Dirac operator in Subsection ~\ref{subsection-spec-dirac}.
     \begin{rem}  Another representation of ~\eqref{eqexactUm}
      is  given by (see ~\cite{balslev})
     \be\label{eqfoldyhelffer} U_m(\xi)=\exp\Big\{-\beta (\alpha\cdot\xi)\theta_m(|\xi|)\Big\},\ee
     $$\theta_m(t)=(2t)^{-1}\arctan (m^{-1}t),\quad m,\,t>0.$$
     Observe that due to the double multiplicity of $\lambda_\pm(\xi)$ the diagonalizing matrix $U_m(\xi)$ is not unique.

     \end{rem}

\subsection{\textbf{EXAMPLE: THE FREE MAXWELL OPERATOR}}\label{subsec-maxwell}

    Consider a pair of three-dimensional vector functions $E$ (the electric field) and $B$ (the magnetic field).  We shall denote by $\binom{E}{B}$ the six-component (column) vector that consists of the vertical arrangement of $E,\,B.$

  The free (vacuum) Maxwell operator $L_{maxwell}$ is a $6\times 6$ self-adjoint operator acting on the combined vector  $\binom{E}{B}$  in
$\Lcal^2(\RT,\mathbb{C}^6)=\Lcal^2(\RT)\bigotimes \mathbb{C}^6,$

\begin{equation} \label{eq:maxwell0}
L_{maxwell}\binom{E}{B}=\frac{1}{i}\begin{pmatrix}
       O_3& -\text{curl} \\
       \text{curl} & O_3
     \end{pmatrix}\binom{E}{B}
=\begin{pmatrix}
       O_3& -D \times \\
       D \times & O_3
     \end{pmatrix}\binom{E}{B},
\end{equation}
where as above $D=(D_1, \, D_2, \, D_3)=\frac{1}{i} \nabla_x,\quad x\in \RT$, and  $O_3$ is the $3\times 3$ zero matrix.

\begin{rem} Recall that the electric ($E$) and magnetic ($B$) fields in vacuum  are divergence-free fields, namely, there are no electric or magnetic charges. This restriction is imposed on the initial data and then the time-dependent equations ensure that the fields evolve as divergence-free fields.

However, here we take the point-of-view that $L_{maxwell}$ acts on the entire space $\Lcal^2(\RT,\mathbb{C}^6).$
\end{rem}

For the Maxwell  operator $L_{maxwell}$ in (\ref{eq:maxwell0}),
the matrix symbol $M_{maxwell}(\xi)$  is readily seen to be of the form
\begin{equation} \label{eq:maxwell1}
M_{maxwell}(\xi)=\begin{pmatrix}
      O_3& - \gamma (\xi)\\
       \gamma (\xi)& O_3
     \end{pmatrix},
\end{equation}
where, for $\xi=(\xi_1,\xi_2,\xi_3),$
\begin{equation*}
\gamma(\xi)=
\begin{pmatrix}
 0 &-\xi_3 & \xi_2  \\
\xi_3 &   0  & -\xi_1 \\
-\xi_2  & \xi_1  &  0
\end{pmatrix}.
\end{equation*}
The eigenvalues of $M_{maxwell}(\xi),\,\xi\neq 0,$ are
\begin{equation*}
\lambda_+(\xi)=|\xi| > \lambda_0(\xi) =0 > \lambda_{-}(\xi)=-|\xi|,
\end{equation*}
 and the
multiplicity of each of them
is two. Since $\gamma (\xi)\xi=0,$ it follows that the kernel of $M_{maxwell}(\xi)$ is two-dimensional, with basis
 vectors
$$\left(
            \begin{array}{c}
              0 \\
              0 \\
              0 \\
              \xi_1 \\
              \xi_2 \\
              \xi_3 \\
            \end{array}
          \right)\,\,,\,\,\,\,\,\left(
            \begin{array}{c}
              \xi_1\\
              \xi_2\\
              \xi_3\\
              0\\
              0 \\
              0 \\
            \end{array}
          \right).$$

Thus the  kernel   is two-dimensional,
independently of $\xi \in {\mathbb R}^3 \setminus \{ 0 \}$. It follows that this operator is uniformly propagative (Definition ~\ref{defn-uniform-prop} below).

 In order to construct the full system of eigenvectors of $M_{maxwell}(\xi)$ we introduce (for any $\xi\neq0$) the following vectors in $\mathbb{C}^3:$

$$\mathfrak{a}(\xi)=\begin{cases}\left(
            \begin{array}{c}
              -\xi_1\xi_2 \\
              \xi_1^2+\xi_3^2 \\
              -\xi_2\xi_3 \\
                          \end{array}\right)\\ \\
                         \left( \begin{array}{c}
              \xi_2 \\
              0\\
              \xi_2\\
                          \end{array}\right)\end{cases},\quad \mathfrak{b}(\xi)=\begin{cases}\left(
            \begin{array}{c}
              -\xi_3|\xi| \\
              0 \\
              \xi_1|\xi| \\
                          \end{array}\right),\,\,\xi_1^2+\xi_3^2>0,\\ \\
                         \left( \begin{array}{c}
              |\xi_2| \\
              0\\
              -|\xi_2| \\
                          \end{array}\right),\,\,\xi_1=\xi_3=0\end{cases}.$$
                           It is readily seen that the vectors $\mathfrak{a}(\xi),\,\mathfrak{b}(\xi),\,\xi$ are mutually orthogonal in $\mathbb{C}^3.$ They also form a right-hand orthogonal system in $\RT,$ and satisfy the relations
      $$\gamma(\xi) \mathfrak{a}(\xi)=|\xi| \mathfrak{b}(\xi),\quad   \gamma(\xi) \mathfrak{b}(\xi)=-|\xi| \mathfrak{a}(\xi) . $$
      In the following corollary, we write a (column) vector in $\mathbb{C}^6$ as a column with two (column) $\mathbb{C}^3$ components.
      \begin{cor}\label{corUpsilonpm0}

       Consider the following three pairs of vectors in $\mathbb{C}^6$
     \be\aligned \Upsilon_+= \set{\frac{1}{\sqrt{|\mathfrak{a}(\xi)|^2+|\mathfrak{b}(\xi)|^2}}\left(
            \begin{array}{c}
              \mathfrak{a}(\xi) \\
                            \mathfrak{b}(\xi) \\
                          \end{array}\right),\,\frac{1}{\sqrt{|\mathfrak{a}(\xi)|^2+|\mathfrak{b}(\xi)|^2}}\left(
            \begin{array}{c}
              -\mathfrak{b}(\xi) \\
                            \mathfrak{a}(\xi) \\
                          \end{array}\right)},
                        \\  \\\Upsilon_{0}= \set{\frac{1}{|\xi|}\left(
            \begin{array}{c}
              0 \\
                            \xi \\
                          \end{array}\right),\,\frac{1}{|\xi|}\left(
            \begin{array}{c}
              \xi \\
                           0 \\
                          \end{array}\right)}\hspace{100pt},\\ \\
                           \Upsilon_{-}= \set{\frac{1}{\sqrt{|\mathfrak{a}(\xi)|^2+|\mathfrak{b}(\xi)|^2}}\left(
            \begin{array}{c}
              -\mathfrak{a}(\xi) \\
                            \mathfrak{b}(\xi) \\
                          \end{array}\right),\,\frac{1}{\sqrt{|\mathfrak{a}(\xi)|^2+|\mathfrak{b}(\xi)|^2}}\left(
            \begin{array}{c}
              \mathfrak{b}(\xi) \\
                            \mathfrak{a}(\xi) \\
                          \end{array}\right)}.
                          \endaligned\ee
                          Then the set of six vectors  $\Upsilon_{+}\cup\Upsilon_{0}\cup\Upsilon_{-}\subseteq\mathbb{C}^6$ constitutes an orthonormal basis of eigenvectors of $M_{maxwell}(\xi).$

                          The pairs $\Upsilon_{\pm }$ (resp. $\Upsilon_{0}$) are eigenvectors associated with $\lambda_{\pm }(\xi)$ (resp. $\lambda_0(\xi)$).

      \end{cor}
      \begin{rem}\label{remsymmmaxwell} The double multiplicities of the eigenvalues implies that the basis vectors in each subspace are not uniquely determined. This explains the apparent asymmetry (with respect to $(\xi_1,\xi_2,\xi_3)$) of the eigenvectors in each of the sets $\Upsilon_{\pm},\,\Upsilon_{0}.$ The operator (and its symbol $M_{maxwell}(\xi)$) is clearly symmetric with respect to orthogonal rotations in $\RT.$
      \end{rem}
      \begin{rem}\label{remTETM}(\textbf{propagation modes}) Let $\binom{E_0(x)}{B_0(x)}$ be a vector function (with values in $\mathbb{C}^6$) whose Fourier transform $ \binom{\widehat{E_0}(\xi)}{\widehat{B_0}(\xi)}\in span \{\Upsilon_{\pm }\},$ for every $0\neq \xi\in\RT.$  In other words,  $ \binom{\widehat{E_0}(\xi)}{\widehat{B_0}(\xi)}$ is orthogonal to $ker(M_{maxwell}(\xi))$ for all $\xi\neq 0.$ Then in particular
          $$<\widehat{E_0}(\xi),\xi>_{\CC^3}=0,\quad <\widehat{B_0}(\xi),\xi>_{\CC^3}=0, \quad \xi\in\RT.$$

          The propagation of the initial data $\binom{E_0(x)}{B_0(x)}$ by the time-dependent Maxwell system yields the solution
 $\binom{E(x, \, t)}{B(x, \, t)}.$  Both $E(x, \, t)$ and $B(x, \, t)$ are superpositions of the plane waves
 $e^{i(<x, \, \xi> \pm t |\xi|)}{\mathfrak a}(\xi)$
 and
 $e^{i(<x, \, \xi> \pm t |\xi|)}{\mathfrak b}(\xi)$, each of which is transverse (in fact  \color{black}orthogonal) to
 the propagation directions $\xi$ or $-\xi$.

   The terminology of ``TE-modes'' (resp. ``TM-modes''), introduced by Lord Rayleigh in 1897, is used to characterize such fields.

   The behavior of these waves for large time is further studied below in Subsection ~\ref{subsection-free-maxwell-spacetime}.

     \end{rem}
%
%
%

 The diagonalization procedure of the Dirac operator (see ~\eqref{eqdiagHm}) can now be repeated, with suitable modifications, for the Maxwell operator as follows.

        We define a continuous map
     $$\xi\hookrightarrow V_0(\xi),\quad \xi\in\RT\setminus \{ 0\},$$
    where   $V_0(\xi)$  is unitary for all $\xi\in\RT\setminus \{ 0\},$ so that

    \be\label{eqmaxcoerciv}
      V_0(\xi)^*M_{maxwell}(\xi)V_0(\xi)=\begin{pmatrix}
       \lambda_+(\xi)I_2& O_2&O_2\\
       O_2&O_2&O_2\\
       O_2&O_2& \lambda_{-}(\xi)I_2\\
     \end{pmatrix}.
    \ee
    We can clearly assume that $V_0(\xi)$ is homogeneous of order zero;\,\,$V_0(\beta\xi)=V_0(\xi),\,\beta>0.$

  \begin{rem}\label{remcontV0} Continuing Remark ~\ref{remsymmmaxwell}, the fact that the eigenvalues are ``separated'' and are of constant multiplicity implies ~\cite[Chapter II.1.4]{kato} that the projections on the eigenspaces are continuous (and indeed real-analytic), hence this is true also for the map
  $$\xi\hookrightarrow V_0(\xi),\quad \xi\neq 0.$$
  Remark that the choice of the
  unitary matrix is not unique. In particular, taking the matrix whose columns   are  the six vectors $\Upsilon_{+}\cup\Upsilon_{0}\cup\Upsilon_{-},$ yields a unitary matrix that diagonalizes $M_{maxwell}(\xi),$ but is not continuous (as a function of $\xi$).
  \end{rem}

  For  $f\in \Lcal^2(\RT,\mathbb{C}^6)$ we define the transformation
    \be\label{eqdefineT0}
    (\mathcal{T}_0f)(\xi)=V_0^\ast(\xi)\widehat{f}(\xi),\,\,\xi\in\RT\setminus\set{0}.
        \ee
        Then
$$\mathcal{T}_0:\Lcal^2_x(\RT,\mathbb{C}^6)\to \Lcal^2_{\xi}(\RT,\mathbb{C}^6)\,,$$
   is unitary and diagonalizes $L_{maxwell}$ in the sense that
  \be\label{eqdiagMaxwell}
   \Big(\mathcal{T}_0L_{maxwell}\mathcal{T}_0^{-1}\widehat{f}\Big)(\xi)=\begin{pmatrix}
        \lambda_+(\xi)I_2& O_2&O_2\\
       O_2&O_2&O_2\\
       O_2&O_2& \lambda_{-}(\xi)I_2\\
     \end{pmatrix}\widehat{f}(\xi),\,\,\xi\in \RT,
  \ee
   is a multiplication by a diagonal matrix.
   \begin{rem} Note that the diagonalization equation applies only to $f$ in the domain of $L_{maxwell}.$ This will be discussed below, in Subsection ~\ref{subsection-spec-maxwell}.
   \end{rem}


    \section{\textbf{WEIGHTED SOBOLEV SPACES and BASIC NOTATION}}
       Here we introduce the weighted Sobolev spaces which play a
       basic role in our treatment. We shall only need these spaces in the framework of $\Lcal^2.$

        For $s \in \RR$ and
     $p$ a nonnegative integer we define:

    \be\label{eqdefineL2s}
     \Lcal^{2,s}(\Rn):=\{ u(x) \quad / \quad \|u \|_{0,s}^2=\int
     \limits_{\Rn} (1+|x|^2)^s |u(x)|^2 dx < \infty \},
    \ee

      \begin {equation}
     \Hcal^{p,s}(\Rn):=\{ u(x) \quad / \quad \nabla^{\alpha} u \in \Lcal^{2,s}, \quad |\alpha| \leq p, \quad
     \|u\|_{p,s}^2= \sum \limits_{|\alpha| \leq p} \|\nabla^{\alpha} u\|_{0,s}^2  \}.
     \end{equation}

     The scalar product in $\Lcal^{2,s}(\Rn)$ is
     \be\label{eqdefinescal0s}
     (u,v)_{0,s}=\int
     \limits_{\Rn} (1+|x|^2)^s u(x)\overline{v(x)} dx .
     \ee

        We write $\Lcal^2$ for $\Lcal^{2,0}$ and $\|u\|_0= \|u
         \|_{0,0}$.  We also write $\Hcal^p$ for $\Hcal^{p,0}$ and , when needed for clarity, $\|u\|_{\Hcal^p}= \|u
         \|_{p,0}$. The scalar product is then denoted by $(\cdot,\cdot).$ We do not distinguish in this
          scalar product between scalar and vector-valued functions; that will
         be clear from the context.

        Let $K>0$ be an integer.  For functions
      valued in $\CC^K,$ if needed for clarity, we denote, as has already been done above for $\Lcal^2$,

       $$\Hcal^{p,s}(\Rn,\mathbb{C}^K)=\Hcal^{p,s}(\Rn)\bigotimes \mathbb{C}^K.$$

       For negative (integer) indices\, $-p$ \, we denote by $\big\{ \Hcal^{-p,s},\quad \|\cdot\|_{-p,s}\big\}$
       the dual space of
        $\Hcal^{p,-s}.$  In particular, observe that any function $f\in \Hcal^{-1,s}$  can be represented (not
              uniquely) as
             \begin{equation}
            \label{h-1s}
              f=f_0+\sum\limits_{k=1}^n i^{-1}\frac{\partial}{\partial x_k}f_k,
              \qquad f_k\in \Lcal^{2,s},\quad 0\leq k\leq n.
              \end{equation}

               On the side of the (Fourier) transformed functions we shall need Sobolev spaces of any real order. The Fourier transform ~\eqref{fourier} carries weighted-$L^2$ spaces to Sobolev spaces (and vice-versa).

       We  let $\Hcalth=\Hcal^{\theta}(\Rn),\,\,\theta\in \mathbb{R}$ be the Sobolev
         space (based on $L^2(\Rn)$) of order $\theta$ of functions obtained as Fourier transforms, namely,
        \begin{equation}\label{eqdefWSobolev}
          \Hcalth= \{ \widehat{u} \quad / \quad u \in \Lcal^{2, \theta}(\Rn),
         \quad \|\widehat{u}\|_{\Hcalth}=\|u\|_{0, \theta} \}.
         \end{equation}
            Of course for an integer $\theta=p$ the definitions of $\Hcal^p$ and $\Hcalth$ are consistent, noting that the latter is used in the Fourier $\xi-$space.

         Remark that the more general (weighted)  spaces $\Hcal^{\theta,s}$
         can be defined, for example, by interpolating
         between $\Hcal^{p,s}$ with integer values of $p,$ but we
         shall make no use of such (weighted) spaces.

       In our study we rely heavily on the trace lemma for functions in
    $\Hcalth(\Rn)$ ~\cite[Proposition 6.3]{BA2},
    \begin{lem}\label{tracelem}  Let $\widehat{h}\in \Hcalth(\Rn),\,\,n\geq 3,\,\,\frac12<\theta<\frac32.$ Denote the sphere of radius $r$ by
    $$S_r=\set{\xi\in\Rn /\,\,|\xi|=r}.$$
    Then
    \be
\int\limits_{S_r}|\widehat{h}|^2d\Sigma_r\leq
C\min(1,r^{2\theta-1})\|\widehat h\|_{\Hcalth}^2,\,\,r>0,
    \ee
    where $C>0$ is independent of $\widehat{h},\,r,$ and  $d\Sigma_r$ is the
    Lebesgue surface measure on $S_r.$

    Furthermore, the family of trace maps $\set{\Phi_r:\Hcalth(\Rn)\to\Lcal^2(S_1),\,r\in\RR_+}$  given by
    $$\Phi_r\widehat{h}(\omega)=r^{\frac{n-1}{2}}\widehat{h}(r\omega),\quad \omega\in S_1,\,\,r\in\RR_+,$$
    is locally H\"{o}lder continuous, in the following sense.

    For a closed interval $[a,b]\subseteq\RR_+$ there exist  constants $C>0$ and $\alpha\in (0,1)$ such that
    \be
    \int\limits_{S_1}\Big|\Phi_{r_2}\widehat{h}(\omega)-\Phi_{r_1}\widehat{h}(\omega)\Big|^2d\Sigma_1\leq C\|\widehat h\|_{\Hcalth}^2|r_2-r_1|^\alpha,\quad r_1,\,r_2\in [a,b].
    \ee
    \end{lem}

    \textbf{Notation.} For any two Banach spaces we use $B(X,Y)$
    to denote the space of linear bounded operators from $X$ to
    $Y,$ equipped with the uniform operator topology.
    \section{\textbf{SPECTRAL STRUCTURE of the UNPERTURBED DIRAC and MAXWELL OPERATORS}}\label{sec-spec-dirac-maxwell}
     In this section, we study in detail the spectral
     properties of the free Dirac and Maxwell operators operators, as introduced in Section ~\ref{unperturbed-sec}.
     This will provide a prelude to the general case that we defer to Section ~\ref{secspecunperturbed}.

    \subsection{\textbf{THE FREE DIRAC OPERATOR}}\label{subsection-spec-dirac}

      We  first consider the spectral density of the free Dirac operator, with suitable estimates in weighted Sobolev spaces. These estimates allow us to derive the Limiting Absorption Principle for this operator (Theorem ~\ref{basiclap} below).

      Recall the definition ~\eqref{eqdefGm} of the transformation $\mathcal{G}_m.$

      Observe that in view of Equation ~\eqref{eqdiagHm} the operator
        $\mathcal{G}_mH_m\mathcal{G}_m^{-1}$ is reduced by the two (complementary) orthogonal subspaces $X_\pm\subseteq \Lcal^2(\RT,\mathbb{C}^4)$ given by
        \be\label{eqsubspacedomainDirac}\aligned
        X_+=\set{f\in \Lcal^2(\RT,\mathbb{C}^4)\,\, / \,\, \mbox{ the two last components of $\mathcal{G}_mf(\xi)$  vanish}},\\
        X_-=\set{f\in \Lcal^2(\RT,\mathbb{C}^4)\,\, / \,\, \mbox{ the two first components of $\mathcal{G}_mf(\xi)$  vanish}}.
        \endaligned \ee

        In order to determine the domain of the self-adjoint operator $H_m$ we repeat the discussion in ~\cite[Section V.5.4]{kato}.

    From ~\eqref{eqdiraccoerciv} we obtain  the coercivity  property (in each reducing subspace)

      \be\label{eq-dirac-coerciv}\aligned|<M_m(\xi)\widehat{f}(\xi),\widehat{f}(\xi)>_{\CC^4}|=
      |<U_m^\ast(\xi)M_m(\xi)U_m(\xi)U_m^\ast(\xi)\widehat{f}(\xi), U_m^\ast(\xi)\widehat{f}(\xi)>_{\CC^4}|\\\geq |\xi||\widehat{f}(\xi)|^2,\quad f\in X_\pm,\,\,\xi\in \RT.
      \endaligned \ee
       It follows that the domain of $H_m$ in each reducing subspace, hence in the full space $\Lcal^2(\RT,\mathbb{C}^4),$ is
     \begin{equation}\label{dom Hm}
  Dom(H_m)=\Hcal^1(\RT,\mathbb{C}^4),
  \end{equation}
     and its spectrum (that is absolutely continuous) is
     \be\label{eqsepecHm}
      spec(H_m)=\RR\setminus(-m,m)
     \ee
      (it is of course $\RR$ if $m=0$).

  Recall that by Equation ~\eqref{eqdiagHm} the operator $H_m$ can be diagonalized  with


       $$\lambda_\pm(r)=\pm\sqrt{r^2+m^2}.$$

    Let $E_m(\lambda)$ be the spectral family associated with
    $H_m.$

     As is
    customary we use $\chi_B$ as the indicator function for a set
    $B\subseteq\RT,$ namely, $\chi_B(\xi)=1$ (resp. $\chi_B(\xi)=0$
    ) if $\xi\in B$ (resp. $\xi\notin B$).

     It is easily seen that if we confine $f\in C^\infty_0(\RT,\mathbb{C}^4)$ then for $\lambda>m$ we have
    \be
    (E_m(\lambda)f,f)=\Big(\begin{pmatrix}
      \chi_{\lambda_+(\xi)\leq \lambda} I_2&O_2\\
       O_2& I_2\\
     \end{pmatrix}\mathcal{G}_mf,\mathcal{G}_mf\Big),
    \ee
    where the right-hand side is the scalar product in $\mathcal{L}^2_\xi(\RT,\mathbb{C}^4).$

    Differentiating the last equality (assuming $f,g$ to be
    sufficiently regular), we get (with $d\Sigma_r$ being the
    Lebesgue surface measure on the sphere of radius $r>0$),
    \be\label{A0}
     \frac{d}{d\lambda}(E_m(\lambda)f,f)=\frac{\lambda}{\sqrt{\lambda^2-m^2}}
     \int\limits_{|\xi|=\sqrt{\lambda^2-m^2}}
|(\mathcal{G}_mf)_+(\xi)|^2d\Sigma_{\sqrt{\lambda^2-m^2}},
    \ee
    where $(\mathcal{G}_mf)_+$ is a $2-$vector consisting of the first
    two components of $\mathcal{G}_mf.$

    An analogous equation is clearly valid in the case
    $\lambda<-m.$

    Since
    $$|(\mathcal{G}_mf)_+(\xi)|\leq
    |(\mathcal{G}_mf)(\xi)|\leq|\widehat{f}(\xi)|,$$
    we conclude from Lemma \ref{tracelem} and the definition ~\eqref{eqdefWSobolev} of the space $\Hcal^s$  that, for any $|\lambda|>m$ and $\frac12<s<\frac32,$ there exists an operator
    $$A_m(\lambda)\in B\Big(\Lcal^{2,s}(\RT,\mathbb{C}^4),\,\Lcal^{2,-s}(\RT,\mathbb{C}^4)\Big),$$
    such that
    \be\label{Aestimate1}\aligned
<A_m(\lambda)f,f>=\frac{d}{d\lambda}(E_m(\lambda)f,f)\hspace{70pt}
\\\leq C\min\Big(\frac{|\lambda|}{\sqrt{\lambda^2-m^2}},\,|\lambda|(\lambda^2-m^2)^{s-1}
\Big)\|\widehat{f}\|_{\Hcal^s}^2
\\=C\min\Big(\frac{|\lambda|}{\sqrt{\lambda^2-m^2}},\,|\lambda|(\lambda^2-m^2)^{s-1}
\Big)\|\widehat{f}\|_{\Hcal^s}^2
    \endaligned\ee
     where $<,>$ is the
     $(\Lcal^{2,-s}(\RT,\mathbb{C}^4),\,\Lcal^{2,s}(\RT,\mathbb{C}^4))$
     pairing.
       \begin{prop}\label{corpropertyA}
       \begin{enumerate}
       \item Let $s>\frac12.$ Then the weak derivative $A_m(\lambda)=\frac{d}{d\lambda}(E_m(\lambda))$ is
     locally bounded and locally
    H\"{o}lder continuous for $|\lambda|>m,$ with respect to the uniform operator topology of
     $B(\Lcal^{2,s}(\RT,\mathbb{C}^4),\Lcal^{2,-s}(\RT,\mathbb{C}^4)).$
     \item Let $s>1.$
    Then the weak derivative $A_m(\lambda)=\frac{d}{d\lambda}(E_m(\lambda))$ is
      uniformly bounded and uniformly
    H\"{o}lder continuous for $\lambda\in\RR,$ with respect to the uniform operator topology of
     $B(\Lcal^{2,s}(\RT,\mathbb{C}^4),\Lcal^{2,-s}(\RT,\mathbb{C}^4)).$

     \end{enumerate}
   \end{prop}
   \begin{proof}
      Note that the topology of $B(\Lcal^{2,s}(\RT,\mathbb{C}^4),\Lcal^{2,-s}(\RT,\mathbb{C}^4))$  becomes weaker as $s$ grows, so without loss of generality we can assume $s<\frac32,$ so the estimate~\eqref{Aestimate1} can be used.

      Consider the expression for the spectral derivative, Equation
       ~\eqref{A0}. Using the expression ~\eqref{eqdefGm}, we can rewrite it explicitly as

      \be\label{eqA0exp}
     <A_m(\lambda)f,f>=\frac{\lambda}{\sqrt{\lambda^2-m^2}}
     \int\limits_{|\xi|=\sqrt{\lambda^2-m^2}}
\Big|[U_m^{\ast}(\xi)\widehat{f}(\xi)]_+\Big|^2d\Sigma_{\sqrt{\lambda^2-m^2}}.
    \ee

        We now estimate the right-hand side of Equation ~\eqref{eqA0exp} in $\Hcal^{\sbar}$ with $\sbar>2.$ By the Sobolev embedding theorem,
        \be\label{eqesttraces3} |\widehat{f}(\xi)|\leq C\|\widehat{f}\|_{\Hcal^{\sbar}},\quad \xi\in\RT.\ee
        Thus
        $$\Big|\int\limits_{|\xi|=\sqrt{\lambda^2-m^2}}
\Big|[U_m^{\ast}(\xi)\widehat{f}(\xi)]_+\Big|^2d\Sigma_{\sqrt{\lambda^2-m^2}}\Big|\leq C(\lambda^2-m^2)\|f\|_{0,\sbar}^2,\,\,\sbar>2,$$
      and from ~\eqref{eqA0exp} we obtain in this case
      \be <A_m(\lambda)f,f>\leq C\lambda(\lambda^2-m^2)^\frac12\|f\|_{0,\sbar}^2,\,\,\sbar>2.
      \ee
        Furthermore, using the explicit form ~\eqref{eqexactUm} of $U_m(\xi)$ and the Sobolev embeddding theorem we have, in addition to  ~\eqref{eqesttraces3}, also
        $$|\nabla_\xi[U_m^{\ast}(\xi)\widehat{f}(\xi)]|\leq C\|\widehat{f}\|_{\Hcal^s},\quad \xi\in\RT,\,\,\sbar>3.$$
       Rewriting  ~\eqref{eqA0exp} in the form
   \be  <A_m(\lambda)f,f>=\frac{\lambda}{\sqrt{\lambda^2-m^2}}
     \int\limits_{|\omega|=1}
\Big|[U_m^{\ast}(\sqrt{\lambda^2-m^2}\omega)\widehat{f}(\sqrt{\lambda^2-m^2}\omega)]_+\Big|^2(\lambda^2-m^2)d\Sigma_{1},
    \ee
    we can differentiate to obtain, for $|\lambda|>m,$
     \be \Big|\frac{d}{d\lambda}<A_m(\lambda)f,f>\Big|\leq C(\lambda^2-m^2)^{-\frac12}\|f\|_{0,\sbar}^2,\,\,\sbar>3.
     \ee
     Since the right-hand side is uniformly locally integrable in $\lambda\in\RR\setminus(-m,m),$ we conclude that the operator-valued function
     $$ A_m(\lambda)\in B(\Lcal^{2,\sbar}(\RT,\mathbb{C}^4),\Lcal^{2,-\sbar}(\RT,\mathbb{C}^4)),\quad |\lambda|\geq m,\,\,\sbar>3,$$
     is uniformly  H\"{o}lder  continuous (and vanishes at $|\lambda|=m$).

    Interpolating this estimate with the boundedness estimates ~\eqref{Aestimate1} we conclude that
    \begin{itemize}
    \item For $s>\frac12$  the operator-valued function
     $$ A_m(\lambda)\in B(\Lcal^{2,s}(\RT,\mathbb{C}^4),\Lcal^{2,-s}(\RT,\mathbb{C}^4)),\quad |\lambda|>m,$$
     is locally bounded and locally  H\"{o}lder  continuous .
     \item  For $s>1$  the operator-valued function
     $$ A_m(\lambda)\in B(\Lcal^{2,s}(\RT,\mathbb{C}^4),\Lcal^{2,-s}(\RT,\mathbb{C}^4)),\quad |\lambda|\geq m,$$
     is uniformly bounded and uniformly H\"{o}lder  continuous in $\lambda\in\RR$ (and vanishes for $\lambda\in [-m,m]$).
    \end{itemize}
    This concludes the proof of the proposition.
   \end{proof}
     We can state a slightly more general fact by taking the norms
     of $f$ and $g$ below in different weighted spaces. In fact ,
     suppose that
      $f,g$  are smooth and
                  compactly supported. Since the bilinear form $<A_m(\lambda)\cdot,\,\cdot>$ is nonegative, applying the Cauchy-Schwarz inequality yields, for any $|\lambda|>m$ and $s,\,l>\frac12,$

                   \be\label{Amfgest} \aligned \qquad \big\vert \frac{d}{d\lambda}(E_m(\lambda)f,g) \big\vert=\big\vert<A_m(\lambda)f,g>\big\vert \leq
          \qquad <A_m (\lambda) f, f >^{\frac12} \cdot < A_m (\lambda) g, g >^{\frac12}
        \\ \leq  C\min\Big(\frac{|\lambda|}{\sqrt{\lambda^2-m^2}},\,
         |\lambda|(\lambda^2-m^2)^{\frac{s+l}{2}-1}\Big)\|\widehat{f}\|_{\Hcal^s}
         \|\widehat{g}\|_{\Hcal^l}
         \\ =  C\min\Big(\frac{|\lambda|}{\sqrt{\lambda^2-m^2}},\,
         |\lambda|(\lambda^2-m^2)^{\frac{s+l}{2}-1}\Big)\|f\|_{0,s}
         \|g\|_{0,l}.
         \endaligned\ee
         It is obvious that in the  inequality above, the
         first $<,>$ is the
     $(\Lcal^{2,-s}(\RT,\mathbb{C}^4),\,\Lcal^{2,s}(\RT,\mathbb{C}^4))$
     pairing, while the second is the
     $(\Lcal^{2,-l}(\RT,\mathbb{C}^4),\,\Lcal^{2,l}(\RT,\mathbb{C}^4))$
     pairing.

The general theory ~\cite[Section 3]{BA2} now yields the \textbf{Limiting Absorption Principle} (LAP) for the unperturbed Dirac operator.
    \begin{thm}\label{basiclap}
    Let $R_m(z)=(H_m-z)^{-1},\,\,Im z\neq 0.$ For any $s,l>\frac12$
    the limits
    \be\label{eqLAPDirac}
    R_m^\pm(\mu)=\lim\limits_{\eps\downarrow 0}R_m(\mu\pm
    i\eps),\,\,\mu\in\RR\setminus[-m,m],
    \ee
    exist in the uniform operator topology of
    $B(\Lcal^{2,s}(\RT,\mathbb{C}^4),\Hcal^{1,-l}(\RT,\mathbb{C}^4)).$

        If $s,l>1$ then the limits in ~\eqref{eqLAPDirac} exist for all $\mu\in\RR,$ or, otherwise stated, they are continuous across the thresholds at $\mu=\pm m.$

         Furthermore, in both cases the limit functions $R_m^\pm(\mu)$ are locally bounded and locally H\"{o}lder
    continuous (in their respective domains) with respect to the uniform operator topology.
    \end{thm}
        \begin{proof} As already pointed out, the properties of $A_m(\lambda),$ as given in Proposition ~\ref{corpropertyA},
          enable us to invoke the general theory and obtain the result in the operator setting of  $B(\Lcal^{2,s}(\RT,\mathbb{C}^4),\Lcal^{2,-l}(\RT,\mathbb{C}^4)).$

          In order to complete the proof we need to show that it is possible to replace $\Lcal^{2,-l}(\RT,\mathbb{C}^4)$ by
          $\Hcal^{1,-l}(\RT,\mathbb{C}^4).$

           Take $f\in \Lcal^{2,s}(\RT,\mathbb{C}^4),$ so that by the already established result, the limit
           $$ R_m^\pm(\mu)f=\lim\limits_{\eps\downarrow 0}R_m(\mu\pm
    i\eps)f,\,\,\mu\in\RR\setminus[-m,m],$$
    exists in $\Lcal^{2,-l}(\RT,\mathbb{C}^4).$

    We have
      $$(H_m-\mu)R_m(\mu\pm
    i\eps)f=f\pm i\eps R_m(\mu\pm i\eps)f,$$
    and since $\set{R_m(\mu\pm i\eps)f,\,\,0<\eps<1}$ is uniformly bounded in $\Lcal^{2,-l}(\RT,\mathbb{C}^4),$
        we obtain the limit (in this space)
       $$\lim\limits_{\eps\downarrow 0}H_mR_m(\mu\pm
    i\eps)f=f+\mu R_m^\pm(\mu)f.$$

        Note that $H_m$ is densely defined and closable in
          $\Lcal^{2,-l}(\RT,\mathbb{C}^4)$ and in fact, in view of the coercivity ~\eqref{eq-dirac-coerciv} its graph norm in this space is equivalent to the norm
           of $\Hcal^{1,-l}(\RT,\mathbb{C}^4).$

%
%

      Retaining the same notation for its closure, we get
       \be\label{eqHmRm=muRm}H_mR_m^\pm(\mu)f=f+\mu R_m^\pm(\mu)f,\ee
       so that indeed
         $$R_m^\pm(\mu)f\in \Hcal^{1,-l}(\RT,\mathbb{C}^4).$$
        \end{proof}

     \begin{rem} \begin{itemize}\item The first part of Theorem ~\ref{basiclap} (namely, $s,\,l>\frac12$ and avoiding the thresholds) was obtained in ~\cite{balslev, yamada}. Both papers made use of the LAP for the Schr\"{o}dinger operator  , by way of formula ~\eqref{squaredirac} .
     \item Note that for the second part of the theorem it suffices to assume
                $s,\,l>\frac12$ with $s+l>2.$
                \item In Proposition ~\ref{Hlap} below we give a somewhat different argument for the proof of ~\eqref{eqHmRm=muRm}, based on the fact that $H_m$ is a constant coefficient operator.
                \end{itemize}
                \end{rem}
     We shall now extend  this theorem to  more general function
     spaces. We take  $s,\,l>1.$

               Let $g \in \Hcal^{1,l}(\RT,\mathbb{C}^4),\,f\in \Hcal^{-1,s}(\RT,\mathbb{C}^4),$
              where  $f$ has a representation of the form
              (\ref{h-1s}), with $f_k\in \Lcal^{2,s}(\RT,\mathbb{C}^4),\,0\leq k\leq 3.$

              Equation ~\eqref{A0} can be extended (at least formally)
              to yield

              \begin{equation}
               \label{A0ex}\aligned
              <A_m(\lambda)[f_0+i^{-1}\sum\limits_{k=1}^3\frac{\partial}{\partial x_k}f_k],\,g>
              \hspace{170pt}\\=\frac{\lambda}{\sqrt{\lambda^2-m^2}}
              \int\limits_{|\xi|^2 =\lambda^2-m^2}<(\mathcal{G}_m\widehat{f_0})_+(\xi)+
              \sum\limits_{k=1}^3\xi_k (\mathcal{G}_m\widehat{f_k})_+(\xi),(\mathcal{G}_m\widehat{g})_+(\xi)>_{\CC^2}
              d\Sigma_{\sqrt{\lambda^2-m^2}},\\ \hspace{100pt} f\in \Hcal^{-1,s},\, g \in \Hcal^{1,l}.\hspace{70pt}\endaligned
              \end{equation}

              Observe that this definition makes good sense even though the
              representation (\ref{h-1s}) is not unique, since
               $$f=f_0+\sum\limits_{k=1}^3i^{-1}\frac{\partial}{\partial x_k}f_k=
               \tilde{f_0}+\sum\limits_{k=1}^3 i^{-1}\frac{\partial}{\partial
               x_k}\tilde{f_k},$$
               implies
              $$ \widehat{f_0}(\xi)+\sum\limits_{k=1}^3\xi_k \widehat{f_k}(\xi)
               =\widehat{\tilde{f_0}}(\xi)+\sum\limits_{k=1}^3\xi_k \widehat{\tilde{f_k}}(\xi)$$
               (as tempered distributions).

               \begin{prop} Equation ~\eqref{A0ex} can indeed be used to define an operator (for which we retain the same notation)
              $$A_m(\lambda)\in B(\Hcal^{-1,s}(\RT,\mathbb{C}^4), \Hcal^{-1,-l}(\RT,\mathbb{C}^4)).$$
              In this setting $<,>$ is  the $(\Hcal^{-1,-l}, \Hcal^{1,l})$
              pairing and  $|\lambda|>m.$
              \end{prop}
              \begin{proof}

               To estimate the operator-norm of $A_m(\lambda)$ as given in
                ~\eqref{A0ex} we use, for $1\leq k\leq 3,$ the estimate ~\eqref{Amfgest}, in the form

                   $$ \aligned |<A_m(\lambda)\frac{\partial}{\partial x_k}f_k,g>|=\frac{|\lambda|}{\sqrt{\lambda^2-m^2}}\,\,\Big|\int\limits_{|\xi|^2 =\lambda^2-m^2}<
               (\mathcal{G}_m\widehat{f_k})_+(\xi),\xi_k(\mathcal{G}_m\widehat{g})_+(\xi)>_{\CC^2}
              d\Sigma_{\sqrt{\lambda^2-m^2}}\Big|\\
                  \\ \leq  C\min\Big(\frac{|\lambda|}{\sqrt{\lambda^2-m^2}},\,
         |\lambda|(\lambda^2-m^2)^{\frac{s+l}{2}-1}\Big)\|\widehat{f_k}\|_{\Hcal^s}
         \|\widehat{\xi_kg}\|_{\Hcal^l}\,\,,
         \endaligned$$
         so from ~\eqref{A0ex} we obtain
         \be\label{A0lamda1}\aligned
         |<A_m(\lambda)f,\,g>|\hspace{150pt}\\\leq  C\min\Big(\frac{|\lambda|}{\sqrt{\lambda^2-m^2}},\,
         |\lambda|(\lambda^2-m^2)^{\frac{s+l}{2}-1}\Big)\|f\|_{-1,s}
         \|g\|_{1,l},\quad  f\in \Hcal^{-1,s}, \,\, g\in \Hcal^{1,l},\,\,s,l>1.
         \endaligned\ee
         \end{proof}

      Theorem ~\ref{basiclap} can now be enhanced to yield
    \begin{prop}
          \label{Hlap}
             The operator-valued function $R_m(z)$ is well-defined
             (and analytic) for nonreal $z$ in the following functional setting.
             \begin{equation}\label{R0stcont}
             z\rightarrow R_m(z) \in
                   B(\Hcal^{-1,s}(\RT,\mathbb{C}^4),
                  \Lcal^{2,-l}(\RT,\mathbb{C}^4)).
                  \end{equation}
                  where $s,\,l>1.$

                   Furthermore, it can be extended
                  continuously from $\mathbb{C}^{\pm}$ to
                  $\overline{\mathbb{C}^{\pm}}$, in this uniform operator
                  topology. The limiting values ( denoted again by
                  $R_m^{\pm}(\lambda)$) are H\"{o}lder continuous
                  in the same topology.

                  The extended function satisfies
                  \begin{equation}
                  \label{H0R0=I}
                  (H_m-z)R_m(z)f=f,\quad f\in \Hcal^{-1,s}(\RT,\mathbb{C}^4),\quad
                   z\in \overline{\mathbb{C}^{\pm}},
                  \end{equation}
                  where for $z=\lambda\in\RR,\quad R_m(z)=R_m^{\pm}(\lambda).$
                \end{prop}
                \begin{proof}

                 By the estimate ~\eqref{A0lamda1} , we get readily $R_m(z)\in B(\Hcal^{-1,s}(\RT,\mathbb{C}^4),
                \Hcal^{-1,-l}(\RT,\mathbb{C}^4))$ if $Im\,z\neq 0,$ as well as the analyticity of the map
                $z\hookrightarrow R_m(z),\,Im\,z\neq 0 .$ Furthermore, the extension
                to $Im\,z=0$ is carried out as in ~\cite[Section 4]{BA1}.

                Equation
                ~\eqref{H0R0=I} is obvious if $Im\,z\neq 0$ and
                $f\in \Lcal^{2,s}.$ By the density of $\Lcal^{2,s}$ in
                $\Hcal^{-1,s}$ , the continuity of $R_m(z)$ on
                $\Hcal^{-1,s}$ and the continuity of $H_m-z$ (in the
                sense of distributions) , we can extend it to all
                $f\in \Hcal^{-1,s}.$

                As $z\to\lambda\pm i\cdot 0$ we have $R_m(z)f\to R_m^{\pm}(\lambda)f$ in
                  $ \Hcal^{-1,-l}.$ Applying the (constant coefficient) operator
                  $H_m-z$ yields, in the sense of distributions,
                  $f=(H_m-z)R_m(z)f\to(H_m-\lambda)R_m^{\pm}(\lambda)f$ which
                  establishes ~\eqref{H0R0=I} also for $Im\,z=0.$

                  Finally,
                  the established continuity of $z\hookrightarrow R_m(z)\in B(\Hcal^{-1,s},
                \Hcal^{-1,-l})$ up to the real boundary and Equation ~\eqref{H0R0=I} imply the continuity of
                the map $z\hookrightarrow H_mR_m(z)\in B(\Hcal^{-1,s}(\RT,\mathbb{C}^4),
                \Hcal^{-1,-l}(\RT,\mathbb{C}^4)).$

               The  stronger continuity claim ~\eqref{R0stcont}  follows since the norm of
                $\Lcal^{2,-l}$ is equivalent to the graph-norm of $H_m$ as a map of $\Hcal^{-1,-l}$
                to itself.
                \end{proof}

%
%

                \begin{rem} Note that we could actually take
                $s,\,l>\frac12$ with $s+l>2.$ This is identical to
                Proposition 2.4 in ~\cite{boussaid}, except that
                here we obtain the H\"{o}lder continuity of the
                limiting values.
                \end{rem}
                \subsection{\textbf{THE FREE MAXWELL OPERATOR}}\label{subsection-spec-maxwell}

                In Subsection ~\ref{subsec-maxwell} we introduced the Maxwell operator $L_{maxwell}.$ It is a constant coefficient differential operator and its symbol $M_{maxwell}(\xi)$ is Hermitian. Thus, it can be realized as a self-adjoint operator in $\Lcal^2(\RT,\mathbb{C}^6),$ for which we retain the same notation.

     The spectrum of the operator is readily seen to be
      \be
      spec(L_{maxwell})=\RR.
      \ee
          The spectrum is absolutely continuous except for the point $\lambda=0.$ In view of Corollary ~\ref{corUpsilonpm0} the eigenspace associated with the zero eigenvalue is given by
          \be\label{eqkerLmaxwell}
          ker(L_{maxwell})=\set{f\in \Lcal^2(\RT,\mathbb{C}^6)\,\,/\,\widehat{f}(\xi)\in span\{\Upsilon_0\},0\neq \xi\in\RT}\,\,.
          \ee
      We  now consider the spectral density and the Limiting Absorption Principle for  the Maxwell operator,  in the setting of weighted Sobolev spaces. The treatment is quite analogous to that of the Dirac operator and we  discuss it briefly, focusing on the  aspects of difference between the two cases. As in the Dirac operator case, we obtain detailed estimates at the threshold energy $\lambda=0.$ Such estimates are needed for global spacetime estimates (Subsection ~\ref{subsection-free-maxwell-spacetime}).
 Note that the Maxwell operator is not elliptic (nor bounded from below), so suitable care is needed with respect to its domain of definition .

 In addition to the kernel expressed in ~\eqref{eqkerLmaxwell} we define  two other
(complementary) subspaces, as follows.

          \be X_\pm=\set{f\in \Lcal^2(\RT,\mathbb{C}^6)\,/\,\,\widehat{f}(\xi)\in span\{\Upsilon_\pm\},0\neq \xi\in\RT}\,\,.
 \ee
    Note (compare ~\eqref{eqsubspacedomainDirac}) that these subspaces can also be expressed as

    \be\label{eqsubspacedomainMax}\aligned
        X_+=\set{f\in \Lcal^2(\RT,\mathbb{C}^6)\,\, / \,\, \mbox{ the four last components of $\mathcal{T}_0f(\xi)$  vanish}},\\
        X_-=\set{f\in \Lcal^2(\RT,\mathbb{C}^6)\,\, / \,\, \mbox{ the four first components of $\mathcal{T}_0f(\xi)$  vanish}},
        \endaligned \ee
    where the transformation $\mathcal{T}_0$ is defined in ~\eqref{eqdefineT0}.

    Observe that these subspaces are those containing  the ``TE, TM'' modes (Remark ~\ref{remTETM}).

     These subspaces are clearly reducing for $L_{maxwell}.$ From ~\eqref{eqmaxcoerciv} we obtain the (partial) coercivity  property
     \be\label{eqMaxpartcoerciv}\aligned|<M_{maxwell}(\xi)\widehat{f}(\xi),\widehat{f}(\xi)>_{\CC^6}|=
     |<M_{maxwell}(\xi)\mathcal{T}_0f(\xi),\mathcal{T}_0f(\xi)>_{\CC^6}|\\\geq |\xi||\mathcal{T}_0f(\xi)|^2=|\xi||\widehat{f}(\xi)|^2,\quad \xi\in\RT,\,f\in X_\pm.\endaligned\ee
        These facts enable us to give an explicit expression for the domain of $L_{maxwell}.$
        \begin{claim}\label{claimdomLmax} The domain of $L_{maxwell},$ as a self-adjoint operator in $\Lcal^2(\RT,\mathbb{C}^6),$ is given  by
     \begin{equation}\label{dom Lmaxwell}
  Dom(L_{maxwell})=ker(L_{maxwell})\oplus ( X_+\cap\Hcal^1(\RT,\mathbb{C}^6))\oplus ( X_-\cap\Hcal^1(\RT,\mathbb{C}^6)).
  \end{equation}
  \end{claim}


    The eigenvalues of the symbol  $M_{maxwell}(\xi)$ are
       $$\lambda_{\pm }(\xi)=\pm |\xi|.$$

    Let $F(\lambda)$ be the spectral family associated with
    $L_{maxwell}.$

    It is easily seen that if we confine $f\in C^\infty_0(\RT,\mathbb{C}^6)$ then for $\lambda>0$ we have
    \be\label{eqmaxwellspecfamily}
    (F(\lambda)f,f)=\Big(\begin{pmatrix}
      \chi_{\lambda_+(\xi)\leq \lambda} I_2&O_2&O_2\\
      O_2&O_2&O_2\\
       O_2& O_2&I_2\\
     \end{pmatrix}\mathcal{T}_0f,\mathcal{T}_0f\Big),
    \ee
 where the right-hand side is the scalar product in $\mathcal{L}^2_\xi(\RT,\mathbb{C}^6).$

    Differentiating the last equality (assuming $f,g$ to be
    sufficiently regular), we get ,
    \be\label{eqA0max}
     \frac{d}{d\lambda}(F(\lambda)f,f)=
     \int\limits_{|\xi|=\lambda}
|(\mathcal{T}_0f)_+(\xi)|^2d\Sigma_{\lambda},
    \ee
    where $(\mathcal{T}_0f)_+$ is a $2-$vector consisting of the first
    two components of $\mathcal{T}_0f.$

    An analogous equation is clearly valid in the case
    $\lambda<0.$

    Since
    $$|(\mathcal{T}_0f)_+(\xi)|\leq
    |(\mathcal{T}_0f)(\xi)|\leq|\widehat{f}(\xi)|,$$
    we conclude from Lemma \ref{tracelem} that, for any $\lambda\in\RR\setminus\set{0}$ and $s>\frac12,$
    \be\label{Aestimate1max}\aligned
\frac{d}{d\lambda}(F(\lambda)f,f)\leq
C\min(1,|\lambda|^{2s-1})\|\widehat{f}\|_{\Hcal^s}^2
\\=C\min(1,|\lambda|^{2s-1})\|f\|_{0,s}^2.
   \endaligned \ee
   It follows that there exists a map
   $$\widetilde{A}(\lambda)\in B(\Lcal^{2,s}(\RT,\mathbb{C}^6),\Lcal^{2,-s}(\RT,\mathbb{C}^6)),$$
   so that
  $$ \frac{d}{d\lambda}(F(\lambda)f,f)= <\widetilde{A}(\lambda)f,f>,$$
     where $<,>$ is the
     $(\Lcal^{2,-s}(\RT,\mathbb{C}^6),\,\Lcal^{2,s}(\RT,\mathbb{C}^6))$
     pairing.
       \begin{prop}\label{propAmaxwellreg}
                   Let $s>\frac12.$
    Then the weak derivative $\widetilde{A}(\lambda)=\frac{d}{d\lambda}(F(\lambda))$ is
      locally bounded and locally
    H\"{o}lder continuous for $\lambda\in\RR\setminus\set{0},$ with respect to the uniform operator topology of
     $B(\Lcal^{2,s}(\RT,\mathbb{C}^6),\Lcal^{2,-s}(\RT,\mathbb{C}^6)).$

     In fact, defining $\widetilde{A}(0)=0,$ The function $\widetilde{A}(\lambda)$ is uniformly bounded and uniformly
    H\"{o}lder continuous for $\lambda\in[-r,r]$ in the operator topology, for every $r>0.$.
        \end{prop}
        \begin{proof}
   The proof is quite parallel to that of Proposition ~\ref{corpropertyA}, so we give a rather brief exposition here. We take $\lambda>0.$

    Consider the expression for the spectral derivative, Equation
       ~\eqref{eqA0max}. Using   ~\eqref{eqdefineT0}, we can rewrite it explicitly as

      \be\label{eqA0expmax}
     <\widetilde{A}(\lambda)f,f>=
     \int\limits_{|\xi|=\lambda}
\Big|[V_0^{\ast}(\xi)\widehat{f}(\xi)]_+\Big|^2d\Sigma_\lambda.
    \ee
         The  local boundedness of    $\widetilde{A}(\lambda)\in B(\Lcal^{2,s}(\RT,\mathbb{C}^6),\Lcal^{2,-s}(\RT,\mathbb{C}^6)),\,\,\lambda\neq 0,$
        is immediate from ~\eqref{Aestimate1max}.

         In what follows $C>0$ is a generic constant that does not depend on $\xi,\,\lambda,\,f.$

        Next consider the right-hand side of Equation ~\eqref{eqA0expmax} and estimate it in $\Hcal^{\sbar},\, \sbar>2.$ In view of the Sobolev embedding theorem (see ~\eqref{eqdefWSobolev} for notation)
        \be\label{eqesttraces3max} |\widehat{f}(\xi)|\leq C\|\widehat{f}\|_{\Hcal^{\sbar}},\quad \xi\in\RT.\ee
        Thus
        $$\Big|\int\limits_{|\xi|=\lambda}
\Big|[V_0^{\ast}(\xi)\widehat{f}(\xi)]_+\Big|^2d\Sigma_{|\lambda|}\Big|\leq C\lambda^2\|f\|_{0,\sbar}^2,\,\,\sbar>2,$$
      and from ~\eqref{eqA0expmax} we obtain in this case
      \be <\widetilde{A}(\lambda)f,f>\leq C\lambda^2\|f\|_{0,\sbar}^2,\,\,\sbar>2.
      \ee
      Using the homogeneity  of $V_0(\xi)$ (of order zero) we have
        $$|\nabla_\xi[V_0^{\ast}(\xi)|\leq \frac{C}{|\xi|}, \quad \xi\in\RT\setminus\set{0}.$$
        In conjunction with  the Sobolev embeddding theorem we have, in addition to  ~\eqref{eqesttraces3max},
        $$|\nabla_\xi[V_0^{\ast}(\xi)\widehat{f}(\xi)]|\leq C\max(1,|\xi|^{-1})\|\widehat{f}\|_{\Hcal^{\sbar}},\quad \xi\in\RT\setminus\set{0},\,\,\sbar>3.$$
       Rewriting  ~\eqref{eqA0expmax} in the form
   \be  <\widetilde{A}(\lambda)f,f>=
     \lambda^2\int\limits_{|\omega|=1}
\Big|[V_0^{\ast}(\lambda\omega)\widehat{f}(\lambda\omega)]_+\Big|^2d\Sigma_{1},
    \ee
    we can differentiate to obtain, for $|\lambda|>0,$
     \be \Big|\frac{d}{d\lambda}<\widetilde{A}(\lambda)f,f>\Big|\leq C\max(\lambda,\lambda^2)\|f\|_{0,\sbar}^2,\,\,\sbar>3.
     \ee
     The right-hand side is  locally integrable in $\lambda\in\RR,$ hence the operator-valued function
     $$\widetilde{A}(\lambda)\in B(\Lcal^{2,\sbar}(\RT,\mathbb{C}^6),\Lcal^{2,-\sbar}(\RT,\mathbb{C}^6)),\quad |\lambda|>0,\,\,\sbar>3,$$
     is locally  H\"{o}lder  continuous (and vanishes at $|\lambda|=0$).

  Interpolating between the local boundedness ~\eqref{Aestimate1max} for $s>\frac12$  and the local H\"{o}lder continuity above (for $\sbar>3$) we obtain the local H\"{o}lder continuity for any $s>\frac12.$
   \end{proof}
   \begin{rem}\label{remA0cont} Since $\lambda=0$ is  an eigenvalue of $L_{maxwell}$  there is no spectral derivative there. However, as seen from the estimate ~\eqref{Aestimate1max}, the weak derivative vanishes as $\lambda\to 0.$
   \end{rem}
     We can state a slightly more general estimate than ~\eqref{Aestimate1max} by taking the norms
     of $f$ and $g$ below in different weighted spaces. In fact ,
     suppose that
      $f,g$  are smooth and
                  compactly supported. Then, as in the derivation of the estimate ~\eqref{Amfgest}, for $\lambda\neq 0,$
                   and $s,l>\frac12,$

                   \be\label{Afgestmax} \aligned \qquad \big\vert \frac{d}{d\lambda}(F(\lambda)f,g) \big\vert \leq
          \qquad <\widetilde{A} (\lambda) f, f >^{\frac12} \cdot < \widetilde{A} (\lambda) g, g >^{\frac12}\\
         \leq  C\min(1,\,
         |\lambda|^{s+l-1}\Big)\|\widehat{f}\|_{\Hcal^s}
         \|\widehat{g}\|_{\Hcal^l}
         \\ = C\min(1,\,
         |\lambda|^{s+l-1}\Big)\|f\|_{0,s}
         \|g\|_{0,l}\,\,.
         \endaligned\ee
         It is obvious that in the  inequality above, the
         first $<,>$ is the
     $(\Lcal^{2,-s}(\RT,\mathbb{C}^6),\,\Lcal^{2,s}(\RT,\mathbb{C}^6))$
     pairing, while the second is the
     $(\Lcal^{2,-l}(\RT,\mathbb{C}^6),\,\Lcal^{2,l}(\RT,\mathbb{C}^6))$
     pairing.

The general theory ~\cite[Section 3]{BA2} now yields the \textbf{Limiting Absorption Principle} (LAP) for the unperturbed Maxwell operator.
    \begin{thm}\label{basiclapmax}
    Let $R_{maxwell}(z)=(L_{maxwell}-z)^{-1},\,\,Im z\neq 0.$ For any $s,l>\frac12$
    the limits
    \be\label{eqLAPmax}
    R_{maxwell}^\pm(\mu)=\lim\limits_{\eps\downarrow 0}R_{maxwell}(\mu\pm
    i\eps),\,\,\mu\in\RR\setminus\set{0},
    \ee
    exist in the uniform operator topology of
    $B(\Lcal^{2,s}(\RT,\mathbb{C}^6),\Lcal^{2,-l}(\RT,\mathbb{C}^6)).$

        The limits in ~\eqref{eqLAPmax} can be extended to $\mu=0,$ or, otherwise stated, they are continuous across the eigenvalue at $\mu=0.$

    Furthermore,  the limit functions $R_{maxwell}^\pm(\mu),\,\mu\in\RR,$ are locally  bounded and locally H\"{o}lder
    continuous  with respect to the uniform operator topology.
    \end{thm}
    \begin{rem}
    Continuing Remark ~\ref{remA0cont}, notice that the limiting values $R_{maxwell}^\pm(0)$ are not, of course, limiting values of the resolvent $(L_{maxwell}\pm i\eps)^{-1}$ as $\eps\to 0.$ However, on the subspace (of ``TE,TM'' modes) orthogonal to the kernel   we have the following corollary.
    \end{rem}

    \begin{cor}\label{corLAPmax} Let  $s,l>\frac12.$ Consider the operator $\mathcal{P}^\bot R_{maxwell}(z)=\mathcal{P}^\bot(L_{maxwell}-z)^{-1},\,\, z=\mu\pm i\eps,$ where  $\mathcal{P}$  is the orthogonal projection on $ker(L_{maxwell})$ in
     $\Lcal^{2}(\RT,\mathbb{C}^6).$ Then the limits

    \be\label{eqLAPorthsubspace}\mathcal{P}^\bot R_{maxwell}^\pm(\mu)=\lim\limits_{\eps\downarrow 0}\mathcal{P}^\bot R_{maxwell}(\mu\pm i\eps)
    ,\,\,\mu\in\RR,
    \ee
    exist in the uniform operator topology of
    $B(\Lcal^{2,s}(\RT,\mathbb{C}^6),\Lcal^{2,-l}(\RT,\mathbb{C}^6)).$

    Furthermore, these limits are H\"{o}lder continuous in the same operator topology.

    \end{cor}
    \begin{proof} Indeed, the corollary follows directly from the last part of Proposition ~\ref{propAmaxwellreg},
                      since the weak derivative $\widetilde{A}(\lambda)=\frac{d}{d\lambda}(F(\lambda))$ (extended by $\widetilde{A}(0)=0$)  is
       uniformly bounded and uniformly
    H\"{o}lder continuous for $\lambda\in[-r,r]$ in the operator topology, for every $r>0.$.
    \end{proof}
     As in the case of the Dirac operator, we can  extend  the theorem to  more general function
     spaces. We continue to assume   $s,\,l>\frac12.$

               Let $g \in \Hcal^{1,l}(\RT,\mathbb{C}^6),\,f\in \Hcal^{-1,s}(\RT,\mathbb{C}^6),$
              where  $f$ has a representation of the form
              (\ref{h-1s}), with $f_k\in \Lcal^{2,s}(\RT,\mathbb{C}^6),\,0\leq k\leq 3.$

              Equation ~\eqref{eqA0max} can be extended
              to yield an operator (for which we retain the same notation)
              \be\widetilde{A}(\lambda)\in B(\Hcal^{-1,s}(\RT,\mathbb{C}^6), \Hcal^{-1,-l}(\RT,\mathbb{C}^6)),\ee
              defined by (where now $<,>$ is used for the $(\Hcal^{-1,-l}, \Hcal^{1,l})$
              pairing and we assume $\lambda>0$),

              $$\aligned
              <\widetilde{A}(\lambda)[f_0+i^{-1}\sum\limits_{k=1}^3\frac{\partial}{\partial x_k}f_k],\,g>
              \hspace{70pt}\\=
              \int\limits_{|\xi|=\lambda}<(\mathcal{T}_0f_0)_1(\xi)+
              \sum\limits_{k=1}^3\xi_k (\mathcal{T}_0f_k)_1(\xi),(\mathcal{T}_0g)_1(\xi)>_{\CC^2}
              d\Sigma_{\lambda},\quad f\in \Hcal^{-1,s},\, g \in \Hcal^{1,l},\endaligned
             $$
              that can be rewritten as

\begin{equation}
               \label{eqA0exmaxa}\aligned
              <\widetilde{A}(\lambda)[f_0+i^{-1}\sum\limits_{k=1}^3\frac{\partial}{\partial x_k}f_k],\,g>
              =
              \int\limits_{|\xi|=\lambda}<(\mathcal{T}_0f_0)_1(\xi),(\mathcal{T}_0g)_1(\xi)>_{\CC^2}d\Sigma_{\lambda}
              \\+\sum\limits_{k=1}^3\int\limits_{|\xi|=\lambda} <(\mathcal{T}_0f_k)_1(\xi)\xi_k,(\mathcal{T}_0g)_1(\xi)>_{\CC^2}
              d\Sigma_{\lambda},\quad f\in \Hcal^{-1,s},\, g \in \Hcal^{1,l}.\endaligned
              \end{equation}

               To estimate the operator-norm of $\widetilde{A}(\lambda)$ in this setting
               we use (\ref{eqA0exmaxa}) and the considerations leading
               to ~\eqref{Afgestmax},
                   for
                $1\leq k\leq 3$
               $$
                \aligned
            &|<\widetilde{A}(\lambda)\frac{\partial}{\partial x_k}f_k,g>|&
                       \\&  \leq
            C\min(1,|\lambda|^{s+l-1})\|f\|_{-1,s}\|g\|_{1,l}&,
            \quad
             f\in \Hcal^{-1,s}, \quad g\in \Hcal^{1,l},
            \endaligned $$
           so that,  instead of ~\eqref{Afgestmax}, we have
           \begin{equation} \label{A0lamda1max}
           |<\widetilde{A}(\lambda)f,g>|\leq
            C\min(1,|\lambda|^{s+l-1})\|f\|_{-1,s}\|g\|_{1,l}.
           \end{equation}
            where $s,l>\frac12.$

            The following proposition is proved in the same way as Proposition ~\ref{propAmaxwellreg}.

            \begin{prop}\label{propAmaxwellrega}
                   Let $s,\,l>\frac12.$
    Then the weak derivative $\widetilde{A}(\lambda)=\frac{d}{d\lambda}(F(\lambda))$ is
      locally bounded and locally
    H\"{o}lder continuous for $\lambda\in\RR\setminus\set{0},$ with respect to the uniform operator topology of
     $B(\Hcal^{-1,s}(\RT,\mathbb{C}^6),
                  \Hcal^{-1,-l}(\RT,\mathbb{C}^6)).$
                  \end{prop}

            When trying to establish the regularization property of the resolvent, in analogy to the Dirac case (Proposition ~\ref{Hlap})  we need to take into account the fact that the kernel is nontrivial, so that regularization can only take part in the subspace $\mathcal{P}^\bot\Lcal^2(\RT,\mathbb{C}^6),$
              where $\mathcal{P}$ is the orthogonal projection on $ker(L_{maxwell})$ in $\Lcal^2(\RT,\mathbb{C}^6),$ as in Corollary ~\ref{corLAPmax}.

      Theorem ~\ref{basiclapmax} can now be enhanced to yield
    \begin{thm}
          \label{Hlapmax}
             The operator-valued function $R_{maxwell}(z)$ is well-defined
             (and analytic) for nonreal $z$ in the following functional setting.
             \begin{equation}\label{R0stcontmax}
             z\rightarrow \mathcal{P}^\bot R_{maxwell}(z) \in
                   B(\Hcal^{-1,s}(\RT,\mathbb{C}^6),
                  \Lcal^{2,-l}(\RT,\mathbb{C}^6)).
                  \end{equation}
                  where $s,\,l>\frac12.$

                   Furthermore, it can be extended
                  continuously from $\mathbb{C}^{\pm}$ to
                  $\mathbb{C}^{\pm}\bigcup\RR$, in this uniform operator
                  topology. The limiting values ( denoted again by
                  $\mathcal{P}^\bot R_{maxwell}^{\pm}(\lambda)$) are locally bounded and locally H\"{o}lder continuous
                  in the same topology.

                  The extended function satisfies, for $z\in \mathbb{C}^{\pm}\bigcup\{\RR\setminus\set{0}\},$
                  \begin{equation}
                  \label{H0R0=Imax}
                  (L_{maxwell}-z)\mathcal{P}^\bot R_{maxwell}(z)f=f-z\mathcal{P} R_{maxwell}(z)f,\quad f\in \Hcal^{-1,s}(\RT,\mathbb{C}^6),\quad
                   ,
                  \end{equation}
                  where for $z=\lambda\in\RR\setminus\set{0},\quad \mathcal{P}^\bot R_{maxwell}(z)=\mathcal{P}^\bot R_{maxwell}^{\pm}(\lambda).$
                \end{thm}
                \begin{rem} Note that the operator $\mathcal{P} R_{maxwell}(z)$ is well-defined for $Im z\neq 0$ and can therefore be extended continuously (in the sense of distributions) to the real axis, in view of the continuity of the left-hand side in Equation ~\eqref{H0R0=Imax}, as is established in the following proof.
                \end{rem}
                \begin{proof}[Proof of the Proposition]
                The proof runs parallel to that of Proposition ~\ref{Hlap}. In fact, in the functional setting of   $B(\Hcal^{-1,s}(\RT,\mathbb{C}^6),
                  \Hcal^{-1,-l}(\RT,\mathbb{C}^6))$ the claims follow from the general theory, in view of Proposition ~\ref{propAmaxwellrega} . Also the proof of  ~\eqref{H0R0=Imax} is identical to that of ~\eqref{H0R0=I}.

                  Since the operator $L_{maxwell}$ is not elliptic, we only need to justify the stronger continuity claim ~\eqref{R0stcontmax}, namely, the fact that $\Hcal^{-1,-l}(\RT,\mathbb{C}^6)$ can be replaced by $\Lcal^{2,-l}(\RT,\mathbb{C}^6)$ in the statement. However,  the restriction of $L_{maxwell}$ to the subspace orthogonal to its kernel is elliptic , as seen from  Equation ~\eqref{dom Lmaxwell}. Therefore , for any $u$ in this subspace, the  graph-norm $ \|u\|_{-1,-l}+ \|L_{maxwell}u\|_{-1,-l}$ is equivalent to the $\Lcal^{2,-l}$ norm $\|u\|_{0,-l}.$

                \end{proof}

            \section{\textbf{ STRONGLY PROPAGATIVE OR ISOTROPIC OPERATORS}}\label{secspecunperturbed}

            We now turn back to the study of the spectral structure of the general (constant coefficient) operator ~\eqref{systemL0}:
            $$L_0=L_{0,hom}+M^0_0=\suml_{j=1}^nM_j^0D_j+M^0_0.$$
            Its ($K\times K$ matrix) symbol is
   given by
      \be\label{strongpropnon}M_0(\xi)=M_{0,hom}(\xi)+M_0^0=
      \suml_{j=1}^nM_j^0\xi_j+M_0^0.
    \ee
     All the common physical systems (Dirac, Maxwell, wave propagation in elastic medium and others) share the basic property of being  \textbf{strongly propagative,} according to the following definition.
            \begin{defn}\label{defn-strong-prop}~\cite{wilcox}:
    The homogeneous operator $L_{0,hom}=\suml_{j=1}^nM_j^0D_j$ is said to be {\upshape{strongly propagative}}
    if $M_{0,hom}(\xi)$
    has a kernel of fixed dimension $0\leq d<K,$   independent of $\xi=(\xi_1,\ldots,\xi_n)\in\Rn\setminus\set{0}.$
  \end{defn}

   The nonzero eigenvalues of $M_{0,hom}$ have the following properties.
     \begin{itemize}
   \item They are positive-homogeneous of degree $1.$
   \item Let $Q_{min}^{M_{0,hom}}(\theta;\xi),\,\theta\in\CC,$ be the minimal polynomial of $M_{0,hom}(\xi).$ Let
    \be\label{eqdefZ}Z=\set{\xi\in \Rn\setminus\{ 0\}\,/\,\mbox{\textbf{  the discriminant of }} Q_{min}^{M_{0,hom}}(\theta;\xi)\,\,\mbox{vanishes}}.\ee
    Then $\Zbar=Z\cup\{0\}$ is a closed cone of  Lebesgue measure zero ~\cite{wilcox} .
     \item In $\Rn\setminus\Zbar$
     every eigenvalue of $M_{0,hom}(\xi)$ has constant multiplicity .
     \end{itemize}
     The distinct nonzero eigenvalues can therefore be enumerated
   as
   \be\label{eqdefevshomog}\mu_\rho(\xi)> \ldots> \mu_1(\xi)>0>\mu_{-1}(\xi)>
   \ldots> \mu_{-\rho}(\xi),\,\,\xi\in\Rn\setminus\Zbar.\ee
   The basic properties of these functions can be summarized as follows.

   \begin{equation}\label{eqevshomogeneous}
   \begin{cases}(i)\,\, \mu_k(\xi) \,\mbox{is continuous on}\,\,\Rn,\,\mbox{and in fact real analytic on}\,\,\Rn\setminus\Zbar,\\
   (ii)\,\,\mu_k(\xi)=-\mu_{-k}(-\xi),\,\,k=1,\ldots,\rho,\,\,\xi\in\Rn\setminus\Zbar,\\
   (iii)\,\,\mu_k(\beta\xi)=\beta\mu_k(\xi),\,\beta>0,\,k=1,\ldots,\rho,\,\,\xi\in\Rn\setminus\Zbar.
   \end{cases}
   \end{equation}
   \be\aligned\label{eqdefprojL0} \mbox{The corresponding projections  are denoted by} \\\,\,\set{P_k(\xi)\,/\,\xi\in \Rn\setminus\Zbar}_{0\neq |k|\leq\rho}.\endaligned\ee
  Remark that for $\xi\in Z$ the disjointness property ~\eqref{eqdefevshomog} is not valid, but the eigenvalues clearly retain the homogeneity property. In fact, they are all bounded on every sphere $|\xi|=r>0,$ as is stated in the following claim.
  \begin{claim}\label{claimbddevshomog}
      There exist constants $c_2>c_1>0$ so that all nonzero eigenvalues satisfy
      \be\label{eqbddevssphere}
         c_1\leq\Big|\mu_{\pm k}\Big(\frac{\xi}{|\xi|}\Big)\Big|\leq c_2,\quad \xi\neq 0,\,\,k=1,\ldots,\rho.
      \ee
  \end{claim}
  \begin{proof} The sphere $|\xi|=1$ is compact, so the boundedness of the nonzero eigenvalues follows from the ``group continuity'' ~\cite[Section II.4]{kato} of these eigenvalues and the assumption that the operator is strongly propagative.

    The inequality ~\eqref{eqbddevssphere} follows from the homogeneity property.
  \end{proof}

             Even though we do not treat in this paper the general nonhomogeneous strongly propagative case, we shall make here a comment concerning its possible eigenvalues.  Suppose then that the constant matrix $M^0_0\neq 0$ is not a scalar matrix.

              Let us  consider the possibility of having an eigenvalue $\eta$ of  $M_0(\xi)$  that is \textit{independent} of $\xi$ for $\xi$ in some open set  $O\subseteq\Rn.$ The existence of such an eigenvalue is equivalent to the fact that $\eta$ is an eigenvalue of
          $L_0.$ In the homogeneous case ($M^0_0=0$) we can only have $\eta=0.$

           For $\eta$ to be such an eigenvalue we need
          $$det(M_0(\xi)-\eta I_K)=0\,,\quad \xi\in O.$$
          Since the determinant is a polynomial in $\xi,$ it follows that  it actually vanishes for all $\xi\in \Rn.$ In particular, $\eta$ is an eigenvalue  of $M^0_0.$

          We conclude that $\eta$ must be contained in the finite set (subset of the set of eigenvalues of $M^0_0$)
          \be\label{eqdefineLambda}\Lambda=\mbox{the set of common eigenvalues of} \quad M^0_0+\suml_{j=1}^na_jM_j^0,\quad \forall a=(a_1,...,a_n)\in\Rn.\ee

       Any further spectral information, in  our approach, requires a detailed study of  the  level surfaces of the eigenvalues (in analogy to the cases of the Dirac operator and the Maxwell system) . The information we need (in the general nonhomogeneous case) requires the use of delicate tools of real algebraic geometry and will not be attempted here. We remark that, to the best of our knowledge, the spectral study of this general class of operators (e.g., the Limiting Absorption Principle) has never been carried out.

        A more restricted class is that of operators for which $Z=\emptyset,$ as follows.
\begin{defn}\label{defn-uniform-prop}~\cite{wilcox1}:
    The operator $L_{0,hom}=\suml_{j=1}^nM_j^0D_j$ is said to be {\upshape{uniformly propagative}}
    if it is strongly propagative and, moreover, the eigenspace associated with every eigenvlaue
    has a constant dimension ,  independent of $\xi\in\Rn\setminus\set{0}.$
  \end{defn}
      For simplicity in what follows we shall refer also to the associated symbols as ``strongly'' or ``uniformly'' propagative.

       J. Rauch studied the asymptotic behavior of solutions of first-order hyperbolic systems, imposing the assumption of a uniformly propagative system ~\cite[Assumption (1.3)]{rauch}. Note that \color{blue}``\color{black}the equations of electromagnetic and elastic waves in crystals are not uniformly propagative. However they are strongly propagative'' ~\cite[Introduction]{weder}.

  In our treatment we shall always assume that $L_{0,hom}$ is strongly propagative.

             We shall restrict our considerations to two classes of operators:
            \begin{itemize}
            \item Strongly propagative \textit{homogeneous}  operators, a generalization of the Maxwell system, as well as the massless Dirac operator.
            \item Nonhomogeneous isotropic operators (see Definition ~\ref{def-spherical-symbol} below), a generalization of the massive Dirac operator.
            \end{itemize}

            \subsection{\textbf{SPECTRAL STRUCTURE OF HOMOGENEOUS STRONGLY PROPAGATIVE OPERATORS}}\label{subsecspecstrprop}

  The assumption that $M^0_0=0$ and $L_{0,hom}$ is strongly propagative permits an explicit representation of its domain as well as  ``partial coercivity'' characterization,  in full analogy to the Maxwell operator (see ~\eqref{eqsubspacedomainMax}).

%
 We define for every index $k=\pm 1,\ldots,k=\pm\rho,$ the subspace

    \be\label{eqsubspacedomainL0hom}
        X_k=\{f\in \Lcal^2(\Rn,\mathbb{C}^K)\,\, / \,\, P_k(\xi)\widehat{f}(\xi)=\mu_k(\xi)\widehat{f}(\xi) ,\,\xi\in \Rn\setminus\Zbar\},\ee
    where the projections $P_k(\xi)$ are as in ~\eqref{eqdefprojL0}.

     These subspaces are clearly reducing for $L_{0,hom}.$ In view of ~\eqref{eqbddevssphere} we obtain the (partial) coercivity  property ,

  $$|<M_{0,hom}(\xi)\widehat{f}(\xi),\widehat{f}(\xi)>_{\CC^K}|\geq c_1|\xi||\widehat{f}(\xi)|^2,\quad f\in X_k,\,1\leq|k|\leq\rho, $$
    (compare the analogous fact ~\eqref{eqMaxpartcoerciv} for the Maxwell system). We therefore conclude that
     \begin{equation}\label{dom L0hom}
  Dom(L_{0,hom})=ker(L_{0,hom})\oplus\suml_{1\leq|k|\leq\rho}\oplus (X_k\cap\Hcal^1(\Rn,\mathbb{C}^K)).
  \end{equation}

\subsubsection{\textbf{The Limiting Absorption Principle for homogeneous strongly propagative systems}}
  \label{subsubsecLAPstrongprop}
       Recall that the set $Z$ was defined in ~\eqref{eqdefZ}.
       Let $ \mu_j(\xi),\,\xi\in \Rn\setminus\Zbar,\,0\neq |j|\leq \rho$ be a nonzero eigenvalue. Let $\lambda\in \RR\setminus\{0\}$ and consider the surface

       \be
       \Gamma_j(\lambda)=\set{\xi\in \Rn\setminus\Zbar\,/\,sgn(j)=sgn(\lambda),\,\mu_j(\xi)=\lambda}.
       \ee
        It is an open smooth submanifold of codimension 1. It is bounded (and bounded away from the origin) in view of Claim  ~\ref{claimbddevshomog}.

        The homogeneity property implies that  the surfaces are homothetic in the sense that
        \be\label{eqgammahomothet}  \Gamma_j(\lambda)=\lambda\Gamma_j(1),\quad \lambda\neq 0.
        \ee
        The surface $\Gamma_j(1)$ plays a basic role in the wave propagation associated with the operator.
        \begin{defn}
     The surfaces $\Gamma_j(1)=\set{\mu_j(\xi)=sgn(j)}$ are called the
     \textbf{slowness surfaces} of the system (see ~\cite[Section 4]{wilcox1}).
  \end{defn}
  The term used in ~\cite{courant} is \textbf{normal surfaces.} We note that treatments by means of global Fourier integral operators necessitate a very careful study of these surfaces as well as very special assumptions on the system (see e.g. ~\cite{liess}).

     Since $\Gamma_j(1)=-\Gamma_{-j}(1),$   we shall henceforth assume $j>0,$ with $\lambda>0.$

        The homogeneity of $\mu_j(\xi)$ implies, by the Euler identity, that
        $$<\xi,\nabla \mu_j(\xi)>_{\Rn}=\mu_j(\xi)=\lambda,\quad \xi\in \Gamma_j(\lambda),$$
        so that the Cauchy-Schwarz inequality and ~\eqref{eqbddevssphere} yield
        \be\label{eqgradmu}
        |\nabla \mu_j(\xi)|\geq c_1>0,\quad \xi\in \Gamma_j(\lambda),\,\,\lambda>0,\,\,1\leq j\leq\rho.
        \ee
        \begin{rem}\label{remstronggroup} The inequality ~\eqref{eqgradmu} means that the ``group velocity'' at the wavefront $\Gamma_j$ is bounded away from zero. Compare with the analogous situation in the study of the asymptotic behavior of solutions of first-order systems ~\cite[Equation (1.7)]{rauch}.
        \end{rem}

        Let $d\Sigma_{\Gamma_j(\lambda)}$ be the Lebesgue measure on $\Gamma_j(\lambda).$ The scaling property ~\eqref{eqgammahomothet} yields
         \be\label{eqmeasurehomothet} d\Sigma_{\Gamma_j(\lambda)}=\lambda^{n-1}d\Sigma_{\Gamma_j(1)}.
         \ee
            Let $\omega\in \Gamma_j(1)$ be a general point, with a corresponding $\lambda\omega\in \Gamma_j(\lambda).$ By ~\eqref{eqmeasurehomothet} the traces of any bounded continuous function $f$ on the two submanifolds satisfy
            \be   \int\limits_{\Gamma_j(\lambda)}|f(\lambda\omega)|^2d\Sigma_{\Gamma_j(\lambda)}=\int\limits_{\Gamma_j(1)}\lambda^{n-1}|f(\lambda\omega)|^2d\Sigma_{\Gamma_j(1)}.
            \ee
        Define (using appropriate scaling) the trace maps of the Sobolev space $\Hcalth(\Rn)$ into $L^2(\Gamma_j(1)),$ by
        \be\label{eqdeftracemap}
        (\Phi^j_\lambda h)(\omega)=\lambda^{\frac{n-1}{2}}h(\lambda \omega),\quad \omega\in \Gamma_j(1),\,\,1\leq j\leq\rho.
        \ee
        To estimate these trace maps we invoke ~\cite[Lemma A.5]{E}.  The uniform lower bound ~\eqref{eqgradmu} implies that the essential condition in that lemma ($|\nabla g|\leq d$) is satisfied \textit{uniformly} for any compact $K\Subset \Gamma_j(1).$ We conclude (by exhaustion) that the estimate can be applied to the smooth manifold $\Gamma_j(1),$ hence these  maps are uniformly bounded for any $\theta>\frac12:$
        \be\label{equnifesttrace}\sup\limits_{\lambda>0}\set{||\Phi^j_\lambda||_{B(\Hcalth(\Rn),\Lcal^2(\Gamma_j(1)))}}<\infty,\ee
        and the operator-valued map $\lambda\hookrightarrow B(\Hcalth(\Rn),\Lcal^2(\Gamma_j(1)))$ is locally h\"{o}lder continuous in the uniform operator topology (compare Lemma ~\ref{tracelem}).

        Let $\set{E_{0,hom}(\lambda),\,\lambda\in\RR}$ be the spectral family of $L_{0,hom}.$
         Since $E_{0,hom}(\lambda)$ commutes with $L_{0,hom},$ it also has a
       symbol , which we denote by $E_{0,hom}(\lambda;\xi).$ The following claim gives an explicit expression in terms of the projections on the eigenspaces (assuming $\lambda>0,$ with a similar expression  for $\lambda<0$). We use $\chi_B$ as the indicator function for a set
    $B\subseteq\Rn,$ namely, $\chi_B(x)=1$ (resp. $\chi_B(x)=0$
    ) if $x\in B$ (resp. $x\notin B$).
       \begin{claim}\label{claimspecfamstrongprop} Let $\lambda>0.$
       Then,

       \be\label{E0lambdaxistrongprop}
                E_{0,hom}(\lambda;\xi)=\suml_{j=1}^{\rho}P_{-j}(\xi)+P_0(\xi)+
                \suml_{j=1}^{\rho}P_j(\xi)\chi_{\set{\mu_j(\xi)\leq\lambda}}.
              \ee
      (refer to Equations ~\eqref{eqevshomogeneous}  and ~\eqref{eqdefprojL0} for notation of eigenvalues and projections).
       \end{claim}

        If $\widehat{f},\widehat{g}\in C^\infty_0(\Rn\setminus\Zbar)$ then, assuming $\lambda>0,$
        we obtain by a well known formula (``coarea formula''~\cite[Appendix C.3]{evans-pde})  for differentiation of volume integrals
        \be\label{eqderivE0hom}
        \frac{d}{d\lambda}(E_{0,hom}(\lambda)f,g)=
        \suml_{j=1}^\rho\int\limits_{\Gamma_j(\lambda)}\frac{<P_j(\xi)\widehat{f}(\xi),P_j(\xi)\widehat{g}(\xi)>_{\CC^K}}{|\nabla \mu_j(\xi)|}d\Sigma_{\Gamma_j(\lambda)},
        \ee
        where $d\Sigma_{\Gamma_j(\lambda)}$ is the Lebesgue surface measure (compare Equation ~\eqref{eqA0max}).

        Note that  the coarea formula requires global Lipschitz condition on $\mu_j(\xi)$.
         However,
          it is obtained by multiplying $\mu_j(\xi)$ by
        a cutoff smooth function  such that $\varphi(\xi)=1$
     on the supports of $\widehat{f}(\xi),\,\,\widehat{g}(\xi)$ and vanishes in a neighborhood of $\Zbar$.
     Then we see from (\ref{E0lambdaxistrongprop})
     that for $\lambda >0$
     \begin{align}\label{E0lambdaxistrongprop+}
     \frac{d}{d\lambda} (E_{0,hom}(\lambda)f,g)
     =
     \suml_{j=1}^\rho
      \frac{d}{d\lambda} \!
       \int_{\Rn} \!\!\!< \!P_{j}(\xi)\chi_{\set{\varphi(\xi)\mu_j(\xi)\leq\lambda}}
   \widehat{f}(\xi),P_{j}(\xi)\widehat{g}(\xi) \!>_{\CC^K}
    \! d\xi .
     \end{align}
    Since  the level set $\{ \xi \, | \,  \varphi(\xi)\mu_j(\xi) = \lambda \,\}$
    is a smooth $(n-1)$-dimensional hypersurface for $\lambda >0$,
      one can apply the coarea formula      to the right  hand side of (\ref{E0lambdaxistrongprop+}), obtaining ~\eqref{eqderivE0hom}.

         The real analyticity of the functions $\mu_j(\xi),\,\xi\in\Rn\setminus\Zbar,$ guarantees that Equation
~\eqref{eqderivE0hom} can be repeatedly differentiated, using higher derivatives of $\widehat{f}(\xi),\,\widehat{g}(\xi).$

            We introduce the subspace of functions permitting such recurrent differentiation in the following definition.
         \begin{defn}\label{defnUpsilon}
        Let $\widehat{\Upsilon^s_Z}$ be the closure of $ C^\infty_0(\Rn\setminus\Zbar,\CC^K)$ in $\Hcal^s(\Rn,\CC^K),$ for $s>\frac12,$ and let $\Upsilon^s_Z\subseteq \Lcal^{2,s}(\Rn,\CC^K)$ be the subspace of its inverse Fourier transforms. It is a closed subspace, equipped with the same norm $\|\cdot\|_{0,s}$ ~\eqref{eqdefineL2s}.
        \end{defn}

 In the following claim we characterize the orthogonal complement of $\Upsilon^s_Z.$
       \begin{claim}\label{claimUpsilonperp} Let $(\Upsilon^s_Z)^\perp\subseteq \Lcal^{2,s}(\Rn,\mathbb{C}^K),\,\,s>\frac12,$ be the orthogonal complement to $\Upsilon^s_Z$ (using the scalar product associated with ~\eqref{eqdefineL2s}). Let $h(x)\in (\Upsilon^s_Z)^\perp.$ Then the Fourier transform of $(1+|x|^2)^sh(x)$ is supported on $\Zbar:$
       \be supp\,\, \mathcal{F}\{(1+|x|^2)^sh(x)\}(\xi)\subseteq\Zbar.
       \ee
       \end{claim}
      \begin{rem}\label{remhinL2s} Note that if $h(x)\in \Lcal^{2,s}(\Rn,\mathbb{C}^K)$ then $(1+|x|^2)^{\frac{s}{2}}h(x)\in \Lcal^{2}(\Rn,\mathbb{C}^K)$ and $(1+|x|^2)^sh(x)\in \Lcal^{2,-s}(\Rn,\mathbb{C}^K).$
      \end{rem}
       \begin{proof}[\textbf{Proof of Claim:}] Let $\psi\in \Lcal^{2,s}(\Rn,\mathbb{C}^K)$ such that
       $\widehat{\psi}\in C^\infty_0(\Rn\setminus\Zbar,\CC^K).$ The scalar product in $\Lcal^{2,s}(\Rn,\mathbb{C}^K)$ can be expressed as
       \be (h,\psi)_{0,s}=\int\limits_{\Rn}<\widehat{h}(\xi),(I-\Delta)^s\widehat{\psi}(\xi)>_{\CC^K}d\xi=
       <(I-\Delta)^s\widehat{h},\widehat{\psi}>,
       \ee
       where $<,>$ in the last term stands for the $\Hcal^{-s}(\Rn,\CC^K),\,\Hcal^s(\Rn,\CC^K)$ pairing. The assumption $h(x)\in (\Upsilon^s_Z)^\perp$ means that
      $ (h,\psi)_{0,s}=0$ hence
      $$supp\,\,(I-\Delta)^s\widehat{h}=supp\,\,\mathcal{F}\{(1+|x|^2)^sh(x)\}\subseteq \Zbar.$$
       \end{proof}
        \begin{rem}\label{rembigZ} Since $Z$ has Lebesgue measure zero (in $\Rn$) it is clear that $\Upsilon^s_Z$ is dense in
        $\Lcal^2(\Rn,\CC^K).$ However, if $Z$ is ``large'' in some $``(n-1)-dimensional''$ sense, then $\Upsilon^s_Z$
        is not necessarily equal to $\Lcal^{2,s}(\Rn,\CC^K).$ In fact, using the terminology of ~\cite{hormander-lions}, $\Upsilon^s_Z$ is equal to
        $\Lcal^{2,s}(\Rn,\CC^K)$ only if $Z$ is \textit{``$s-$polar''}. In other words, if it has zero Bessel potential theoretic capacity of order s ~\cite[Sections 10.4, 13.2]{mazya}.
        \end{rem}
        Given the special algebraic structure of $Z$ (see ~\eqref{eqdefZ}) we introduce the following conjecture.
        \begin{conjecture}\label{conjecdense}
           The set $Z$ is $s-$polar for $s\in(\half,\frac32),$ hence the subspace $\Upsilon^s_Z$
        is  equal to $\Lcal^{2,s}(\Rn,\CC^K).$
        \end{conjecture}
        In conjunction with ~\eqref{eqgradmu} and ~\eqref{equnifesttrace} we conclude from  Equation ~\eqref{eqderivE0hom} that, for any $s>\frac12$ there exists a constant $C>0,$ depending only on $s,c_1,$ so that, for all $\lambda>0,$
        and all $f,\,g\in \Upsilon^s_Z,$
        \be\label{eqderivE0basic}\aligned
        \Big|\frac{d}{d\lambda}(E_{0,hom}(\lambda)f,g)\Big|\leq \frac{1}{c_1}\suml_{j=1}^\rho
        \int\limits_{\Gamma_j(\lambda)}|\widehat{f}(\xi)|\cdot|\widehat{g}(\xi)|d\Sigma_{\Gamma_j(\lambda)}\\=
        \frac{1}{c_1}\suml_{j=1}^\rho
        \int\limits_{\Gamma_j(1)}|\Phi^j_\lambda\widehat{f}(\omega)||\Phi^j_\lambda\widehat{g}(\omega)|
        d\Sigma_{\Gamma_j(1)}\leq C\|f \|_{0,s}\|g \|_{0,s}.
        \endaligned\ee
      Any continuous functional on the closed subspace $\Upsilon^s_Z$ can be uniquely extended to a functional on $\Lcal^{2,s}(\Rn,\CC^K)$ (namely, a function in $\Lcal^{2,-s}(\Rn,\CC^K)$) by defining it as zero on the orthogonal complement.

      It follows (compare Proposition ~\ref{propAmaxwellreg}) that
     \begin{cor}
      There exists a map
   $$\widetilde{A_{0,hom}}(\lambda)\in B(\Upsilon^s_Z,\Lcal^{2,-s}(\Rn,\CC^K)),\,\,s>\frac12,$$

   so that
  \be\label{eqderivE0outZ} \frac{d}{d\lambda}(E_{0,hom}(\lambda)f,f)= <\widetilde{A_{0,hom}}(\lambda)f,f>,\quad f\in \Upsilon^s_Z,\ee
     where $<,>$ is the
     $(\Lcal^{2,-s}(\Rn,\CC^K),\,\Lcal^{2,s}(\Rn,\CC^K))$
     pairing.

     The map $\widetilde{A_{0,hom}}(\lambda)$ is uniformly bounded
     \be\label{eqestAohom}\|\widetilde{A_{0,hom}}(\lambda)\|_{B(\Upsilon^s_Z,\Lcal^{2,-s}(\Rn,\CC^K))}\leq C,\quad \lambda\neq 0,\ee
     and locally H\"{o}lder continuous in the uniform operator topology.
     \end{cor}
        The global uniform boundedness ~\eqref{eqestAohom} will play a crucial role in the spacetime estimates of Section
        ~\ref{secspacetimestrongprop}.


     The general theory (see Theorem ~\ref{th-LAP-ABS} below with $\MX=\Upsilon^s_Z$ )  now yields the LAP in this case as follows.

     \begin{thm}\label{basiclapstrongprop}
    Let $R_{0,hom}(z)=(L_{0,hom}-z)^{-1},\,\,Im \,z\neq 0.$ For any $s>\frac12$
    the limits
    \be\label{eqLAPL0hom}
    R_{0,hom}^\pm(\mu)=\lim\limits_{\eps\downarrow 0}R_{0,hom}(\mu\pm
    i\eps),\,\,\mu\in\RR\setminus\set{0},
    \ee
    exist in the uniform operator topology of
    $B(\Upsilon^s_Z,\Lcal^{2,-s}(\Rn,\CC^K)).$


    Furthermore,  the limit functions $R_{0,hom}^\pm(\mu),\,\mu\in\RR\setminus\set{0},$ are locally  bounded and locally H\"{o}lder
    continuous  with respect to the uniform operator topology.

     \end{thm}
  \begin{rem}\label{remLAPWeder} The LAP result of Theorem ~\ref{basiclapstrongprop} was proved by Weder in ~\cite{weder1,weder} for intervals interior to $\RR\setminus\set{0} ,$ in the operator space $B(\Lcal^{2,s}(\Rn,\CC^K),
                  \Hcal^{1,-s}(\Rn,\CC^K)).$ We emphasize that we do not believe that the presence of the singular set $Z$ could be entirely dismissed. The proof in ~\cite{weder1} relies on the commutator approach, and we were not quite able to follow the details there. On the other hand the proof in ~\cite{weder} is essentially based on the  methodology of trace maps. It seems to us to be fundamentally flawed, and this impression has not changed even after a long correspondence with him. The proof of ~\cite[Theorem A.1]{weder} involves  a deformation map of the slowness surface $\Gamma_j(1)$ onto the unit sphere followed by an application of  the trace theorem on the sphere. Thus, the measure $dw_j$ on $\Gamma_j(1)$ is defined by the radial projection on the unit sphere (see ~\cite[Eq. (A.20)]{weder}) so as to obtain a ``polar decomposition'' $d^nk=\rho^{n-1}d\rho dw_j$~\cite[Eq. (A.19)]{weder}. This is of course wrong, since the coarea formula is ignored. That formula introduces a denominator $|\nabla\mu_j(k)|$ (see Eq.~\eqref{eqderivE0hom}) in the last expression, which is singular on $Z.$ Effectively, he argues that $Z$ is of ``measure zero'' in $\Gamma_j(1),$ so our  $\Upsilon^s_Z$ can be identified with $\Lcal^{2,s}(\Rn,\CC^K)).$ This whole argument is applicable in obtaining a trace on every star-shaped surface, no matter how singular, and this is clearly wrong, see Theorem 2.3 in ~\cite{agmon-hormander} and Remark ~\ref{rembigZ}.
                  \end{rem}

  \subsection{\textbf{SPECTRAL DENSITY OF ISOTROPIC OPERATORS}}\label{subsecspecspherical}.
     The examples of the Dirac and Maxwell operators motivate our next
    definition. In fact, like these two examples, all physical models where there is no ``built in'' preference for specific (spatial) directions, naturally fall into the category of \textbf{isotropic operators} ~\cite[Section 4]{wilcox1}  , that we recall next.

%

    \begin{defn}\label{def-spherical-symbol}
     The operator $L_0$ ~\eqref{systemL0} is said to be{ \upshape{isotropic
     } } if the eigenvalues of its symbol (see ~\eqref{strongpropnon})
     $M_0(\xi)$ are functions of $|\xi|.$
    \end{defn}
        We assume in addition that $L_{0,hom}$ is \textbf{strongly propagative} (Definition ~\ref{defn-strong-prop}).

      If $L_0$ is \textbf{homogeneous, }namely,  $M^0_0\equiv 0,$ then $M_0(\xi)=M_{0,hom}(\xi).$  By definition the zero eigenvalue $\mu_0(\xi)=0$ is of fixed dimension $d_0\geq 0.$ The singular set $Z$ ~\eqref{eqdefZ} is empty and the eigenvalues $\mu_k$ (see Equation ~\eqref{eqdefevshomog}) satisfy
       \be\label{eqdefevshommuc} \mu_k(\xi)=sign(k)\mu_{|k|}^c|\xi|,\quad \xi\neq 0,\,\,\pm k=1,2,...,J,\ee
       where $\mu_k^c$ are positive  constants such that $\mu_J^c>\ldots >\mu_1^c>0.$

       Each eigenvalue $\mu_k$ is of fixed dimension $d_k>0,\,\pm k=1,2,...,J,$ and $d_k=d_{-k}.$ In particular, in this case the operator is \textbf{uniformly propagative}
        (Definition ~\ref{defn-uniform-prop}).

        It is therefore a special case of the class considered above in Subsection ~\ref{subsecspecstrprop} and does not require a further consideration here.

%

     We now turn to the \textbf{nonhomogeneous operator.}

     In this case the eigenvalues are functions of a single variable $r=|\xi|,$ and we denote them by
      \be\label{eqevsisotropicnonhom}\lambda_1(r)\leq\lambda_2(r)\leq\cdots\leq\lambda_\rho(r).\ee
      However now these eigenvalues are not homogeneous functions of $r>0,$ and, unlike the homogeneous case, their multiplicity is not fixed. In other words, two (or more) different eigenvalues $\lambda_i(r)\neq \lambda_j(r)$ can ``coalesce'' at a point $r=r_0.$ Such a point $r_0$ is called `` a crossing point''.

      As we shall see below, it will be necessary to forsake their ordering in order to maintain their analyticity.





          Since $L_{0}$ is  isotropic,  the eigenvalues of
\begin{align*}
M_0(\xi)= M_{0,hom}(\xi) + M_0^0 =|\xi|\sum_{j=1}^n M_j^0 \omega_j+ M_0^0
\quad  (\omega=\xi/|\xi|)\\
\color{blue}
\end{align*}
are functions of $|\xi|.$ In particular,    the eigenvalues of
$M_0(\xi)=M_0(|\xi|\omega)$, together with their multiplicities,
are independent of $\omega \in {\mathbb S}^{n-1}$.

By virtue of this fact, we can  take
$\omega =e_1=(1, \, 0, \, \cdots, \, 0)$, and
study
the eigenvalues of
\begin{equation} \label{eq:1}
M_0(|\xi|e_1)=|\xi| M^0_1 + M_0^0.
\end{equation}

Taking $r=|\xi|>0,$ the eigenvalue study is reduced to the study of the symmetric matrix, depending (linearly) on a positive parameter,
\begin{equation} \label{eq:2}
T(r) = r M_1^0 + M_0^0 , \qquad r >0.
\end{equation}
However, it is useful to regard $T(r)$ as a function of the coordinate $r\in\RR.$ We can now appeal directly to the analytic perturbation theory of Hermitian matrices ~\cite[Section II.6 ]{kato}.

%
%
We conclude that the eigenvalues $\lambda_1(r)$, \, $\cdots$, \, $\lambda_{\rho}(r)$
of $T(r)$ are analytic functions of $r\in\RR.$ The sum of their multiplicities is $K,$ and each of them is constant in intervals not containing crossing points, as will be explained below.

We denote by
\be\label{eqdefPjisotropic} P_1(\xi), \, \cdots, \, P_{\rho}(\xi),\ee
the corresponding projections.

Note that $P_j(\xi)$ cannot be assumed to depend on $|\xi|.$
%
Furthermore, in every closed interval $[\alpha,\beta]\Subset \RR$ there is at most a finite number of
crossing  points. In fact, the fact that $T(r)$ depends \textit{linearly} on $r,$ enables us to claim even more.
  \begin{claim}\label{claimcrossing}
      There is at most a finite number of crossing points in the whole real line $r\in \RR.$
  \end{claim}
  \begin{proof}
      The crossing points are zeros of the discriminant of the minimal polynomial of $T(r),$ as a function of $r\in\RR.$ However, clearly this discriminant is an algebraic function of $r,$ and as such can have at most a finite number of zeros.
  \end{proof}
%
%
%
%
%
%
%
%
%
%
%

             As already noted above,   the analytic perturbation theory of Hermitian matrices ~\cite[Section II.6 ]{kato} implies
that $\lambda_1(r)$, \, $\cdots$, \, $\lambda_{\rho}(r)$
 are analytic functions of $r>0.$  Clearly, to maintain them as analytic functions we cannot order them, as there may be  \textbf{crossing points.}
      We define the \textbf{crossing values} as the finite set $$\set{\lambda\in\RR/\,\lambda=\lambda_j(q),\,\,\,\,\mbox{for some}\,\,j\,\, \mbox{ and some crossing point }\,\,q}.$$

       \begin{rem} In the isotropic case the singular set $Z$ in Equation ~\eqref{eqdefZ} can be expressed  explicitly as follows.

 $$Z=\begin{cases} \emptyset,\quad\mbox{the homogeneous case},\\
 \mbox{a finite union of spheres in the nonhomogeneous case.}
 \end{cases}$$  If a sphere of radius $r>0$ (centered at the origin) is included in $Z,$ then $r$ is a \textbf{crossing point.}

     \end{rem}

             Define the spherical surfaces
             \be\label{eqsphericalGammaj}\Gamma_j(\lambda)=\set{\xi\in \Rn\,/\, \, \lambda_j(|\xi|)=\lambda},\,\, \lambda\in \RR\setminus\Lambda.\ee

              If $\Gamma_j(\lambda)\subseteq Z$ then $\lambda$ is a crossing value.

             However, as noted above, even at crossing points the eigenvalues remain analytic.  Thus, the only values of $\lambda$ to be excluded are the \textbf{critical values,} where, by definition, $\lambda=\lambda_j(r)$  and $\lambda_j'(r)=0$ \,for some $j\in\set{1,2,...,\rho}$\,and some $r\in\RR.$
              \begin{claim}\label{claimfincritev} There is at most a finite number of critical values of the eigenvalues $\lambda_j(r),\,j=1,2,...,\rho.$
        \end{claim}
        \begin{proof} The functions $\lambda_j(r)$ are roots of the algebraic equation $det(T(r)-\lambda I_K)=0 $
         ($T(r)$ is defined in ~\eqref{eq:2}), and
   are therefore algebraic functions of the real variable $r\in\RR.$ We can now apply the classical argument in ~\cite[Section 14.3]{hormander}; the set of critical values is of measure zero by Sard's theorem, hence being semi-algebraic set it must be  finite .
        \end{proof}

              Recall the definition ~\eqref{eqdefineLambda} of the finite set $\Lambda,$ that contains all possible eigenvalues of $L_0.$

             \begin{defn}\label{defLambdaisotr} The set $\Lambda$ of ~\eqref{eqdefineLambda} is extended (retaining the same notation) to include also the finitely many critical values.
             \end{defn}
             \begin{rem}\label{remisotropicgroup}  Continuing Remark ~\ref{remstronggroup}: the requirement that $\lambda$ is not a critical value   means that the ``group velocity'' at the wavefront $\Gamma_j(\lambda)$ is bounded away from zero, for all $j\in\set{1,2,...,\rho}.$
        \end{rem}

%
%
%

       Thus, the only values of the spectral parameter to be avoided are the values in $\Lambda,$ and not necessarily all
       crossing values. Let $\lambda\in\RR\setminus\Lambda,$
%

\subsubsection{\textbf{The Limiting Absorption Principle for isotropic operators}}
  \label{subsubsecLAPisotrop}

        Let $E_0(\lambda)$ be  spectral family  associated with $L_0.$
       Since $E_0(\lambda)$ commutes with $L_0,$ it also has a
       symbol, which we denote by $E_0(\lambda;\xi).$ As in the case of Claim ~\ref{claimspecfamstrongprop} we have here.

%
       \begin{claim}\label{claimspecfamisotropic} Let $\lambda\in\RR\setminus\Lambda.$
       Then, with projections $P_j$ as in ~\eqref{eqdefPjisotropic},
       \be\label{E0lambdaxiiso}
                E_0(\lambda;\xi)=\suml_{j=1}^{\rho
       }P_j(\xi)\chi_{\set{\lambda_j(|\xi|)\leq\lambda}}.
              \ee

       \end{claim}

         \begin{cor}\label{corspecderivisotr}
         If $\hat{f},\,\hat{g}\in C^\infty_0(\Rn,\mathbb{C}^K)$ and $\lambda\in\RR\setminus\Lambda,$ then
         \be\label{eqE0deriv}\aligned
          \frac{d}{d\lambda}(E_0(\lambda)f,g)=\suml_{j=1}^{\rho}\int\limits_{\Gamma_j(\lambda)}
          \frac{<P_j(\xi)\hat{f}(\xi),\hat{g}(\xi)>_{\CC^K}}{|\nabla\lambda_j(\xi)|}\chi_{\set{\lambda_j(|\xi|)=\lambda}}d\Sigma_{r_j}\\=
          \suml_{j=1}^{\rho}|\lambda_j'(r_j)|^{-1}\int\limits_{\Gamma_j(\lambda)}
          <P_j(\xi)\hat{f}(\xi),\hat{g}(\xi)>_{\CC^K}\chi_{\set{\lambda_j(|\xi|)=\lambda}}d\Sigma_{r_j}.
         \endaligned\ee
         with $d\Sigma_{r_j}$ being the
    Lebesgue surface measure on the sphere $\Gamma_j(\lambda)$ of radius $r_j>0$ such that $\lambda_j(r_j)=\lambda$ (see ~\eqref{eqsphericalGammaj}).
         \end{cor}

         \begin{rem} In  Equation ~\eqref{eqE0deriv} we know that $\lambda_j'(r_j)\neq 0$ since $\lambda$ is not a \textit{critical value.}
         \end{rem}
          In analogy to the case of the strongly propagative system (Theorem ~\ref{basiclapstrongprop}) we can state here the LAP as following from the general theory.
       \begin{thm}\label{basiclapisotropic}
    Let $L_0(D)$ be isotropic and $R_{0}(z)=(L_{0}-z)^{-1},\,\,Im \,z\neq 0.$ Then for any $s>\frac12$
    the limits
    \be\label{eqLAPL0isot}
    R_{0}^\pm(\mu)=\lim\limits_{\eps\downarrow 0}R_{0}(\mu\pm
    i\eps),\,\,\mu\in\RR\setminus\Lambda,
    \ee
    exist in the uniform operator topology of
    $B(\Lcal^{2,s}(\Rn,\CC^K),\Lcal^{2,-s}(\Rn,\CC^K)).$


    Furthermore,  the limit functions $R_{0}^\pm(\mu),\,\mu\in\RR\setminus\Lambda,$ are locally  bounded and locally H\"{o}lder
    continuous  with respect to the uniform operator topology.

     \end{thm}

              \section{\textbf{THE CLASS OF PERTURBED OPERATORS}}\label{perturbed-sec}

              We now turn to the study of the spectral structure of perturbations of the operators introduced in Section ~\ref{unperturbed-sec}. Our ultimate goal is the study of the operator $L$ introduced in ~\eqref{systemL}.
              This is done by regarding this operator as a perturbation of $L_0,$ as given in ~\eqref{systemL0}.

              We take the most basic perturbation of the form
               \be
                    L=L_0+V(x),
               \ee
                  where $V(x)$ is a Hermitian $K\times K$ matrix  that decays as $|x|\to\infty.$ In Subsection
                  ~\ref{subsection-spec-general-perturb} we expound the abstract theory of such perturbations. In particular, general conditions are given that imply the absolute continuity (and LAP) of the continuous spectrum, apart possibly from a discrete sequence of embedded eigenvalues.

                  In subsection ~\ref{subsection-spec-dirac-perturb} we apply the general theory to the potential perturbation of the Dirac operator. Observe that this includes the case of the magnetic Dirac operator, as the magnetic field can be merged into the potential perturbation. We choose to deal with this example in detail as it illustrates the applicability of the general theory, and also allows us to give decay conditions on the potential, to the effect that there are only \textit{finitely many} eigenvalues in the spectral gap (of the massive Dirac operator, see Subsubsection ~\ref{secfiniteevs}.

                  The perturbed Maxwell system can also be reduced to the case of potential perturbation ~\cite[Section 1.4]{tamura1}, but we choose not to treat it here in detail, as the paper is already quite long.

                  \subsection{\textbf{GENERAL THEORY: PERTURBATION BY A POTENTIAL }}\label{subsection-spec-general-perturb}.\newline

                  Our treatment is based on the general theory expounded in  ~\cite[Sections 3-4]{ba4}. The abstract setting will allow us to consider various operator settings in a unified way. We briefly recall some definitions and statements that will be needed here.

   Let $H$ be a self-adjoint operator in $\Lcal^2(\Rn,\CC^K).$    The scalar product in $\Lcal^2(\Rn,\CC^K)$ is denoted by
$(\;,\,).$

              Suppose that there exists a Hilbert space such that $\mathcal{X}\subseteq
\Lcal^2(\Rn,\CC^K)$, and the embedding is dense and continuous. In other words,
$\mathcal{X}$ can be considered as a dense subspace of $\Lcal^2(\Rn,\CC^K)$, equipped
with a stronger norm. Then, of course, $\mathcal{X}\hookrightarrow
\Lcal^2(\Rn,\CC^K)\hookrightarrow\mathcal{X^*}$, where $\mathcal{X^*}$ is the
anti-dual of $\mathcal{X}$; the continuous additive functionals $l$ on
$\mathcal{X}$, such that $l(\alpha v)=\overline{\alpha}\,l(v)$, $\alpha\in
\mathbb{C}$.

We use $\|x\|_\MX,\,\|x^*\|_{\MX^*}$ for the norms in $\MX$, $\MX^*$,
respectively, and designate by $<\;,\,>$ the $(\MX^*,\MX)$ pairing.

The (linear) embedding $h\in\Lcal^2(\Rn,\CC^K)\hookrightarrow
h^*\in\mathcal{X^*}$ is obtained as usual by the scalar product (in
$\Lcal^2(\Rn,\CC^K)$), $h^*(x)=(h,x)$.

We introduce still another Hilbert space
$\MX_H^*$, which is a dense subspace of $\MX^*$, equipped with a stronger norm
(so that the embedding $\MX_H^*\hookrightarrow\mathcal{X^*}$ is continuous).
However, we do not require that $\Lcal^2(\Rn,\CC^K)$ be embedded in $\MX_H^*$. As indicated
by the notation, $\MX_H^*$ may depend on $H .$  A typical case would be when $H$ can be extended
as a densely defined operator in $\MX^*$ and $\MX_H^*$ would be its domain
there, equipped with the graph norm.

Let $\set{E(\lambda)}$ be the spectral family of $H$. We denote by $E(B)$  the spectral projection on any Borel
set $B$ (so that $E(\lambda)=E(-\infty,\lambda)$).

\begin{defn}
  \label{defn-typeU}
  Let $U\subseteq \RR$ be open and let $0<\alpha\leq 1$. Assume that $U$ is of
  \textit{full spectral measure\/}, namely, $E(\RR\setminus U)=0$.  Then $H$
  is said to be of \textit{type} $\left(\MX,\MX^*_H,\alpha, U\right)$ if the
  following conditions are satisfied:
  \begin{enumerate}
  \item The operator-valued function
  $$
  \lambda\to E(\lambda)\in B(\MX,\MX^*),\quad \lambda\in U,
  $$
  is weakly differentiable with a locally H\"{o}lder continuous derivative in
  \linebreak $B(\MX,\MX^*_H)$; that is, there exists an operator-valued
  function
  $$
  \lambda \to A(\lambda)\in B(\MX,\MX^*_H),\quad \lambda\in U,
  $$
  so that  \[
  \frac{d}{d\lambda}(E(\lambda)x,y)=\,<A(\lambda)x,y>, \quad x,\,y\in
  \MX,\,\lambda\in U, \] and such that for every compact interval $K\subseteq
  U$, there exists an $M_K>0$ satisfying
  \[
  \left\|A(\lambda)-A(\mu)\right\|_{B(\MX,\MX^*_H)}\leq
  M_K\left|\lambda-\mu\right|^\alpha,\quad \lambda,\,\mu\in K.
  \]
  \item For every bounded open set $J\subseteq U$ and for every compact interval
  $K\subseteq J$, the operator-valued function (defined in the weak sense)
  $$
  z\to\int_{U\setminus J}\frac{A(\lambda)}{\lambda - z}\,d\lambda,\quad
  z\in\mathbb{C},\, \Re z\in K,\,|\Im z|\leq 1,
  $$
  takes values and is H\"{o}lder continuous in the uniform operator topology
  of $B(\MX,\MX^*_H)$, with exponent $\alpha$.
  \end{enumerate}
\end{defn}

We can now state the basic theorem, concerning the \textit{Limiting Absorption Principle} (LAP) in this
setting. We use the notation $\mathbb{C}^\pm=\set{z\in\mathbb{C},\,\,\pm \Im\,z>0},$ and denote by $R(z)=(H-z)^{-1},\,\,z\in\mathbb{C}^\pm\,,$ the resolvent of $H.$

\begin{thm}\label{th-LAP-ABS}
  Let $H$ be of\/ \textup{type} $\left(\MX,\MX^*_H,\alpha, U\right)$
  \textup{(}where $U\subseteq \RR$ is open and $0<\alpha\leq 1$\textup{)}.
  Then $H$ satisfies the\/ \textup{LAP} in $U$. More explicitly, the limits
  \[
  R^\pm(\lambda)=\lim_{\eps\downarrow 0}R(\lambda\pm i\eps),\quad \lambda \in
  U,
  \]
  exist in the uniform operator topology of $B(\MX,\MX^*_H)$ and the extended
  operator-valued function
  $$
  R(z)=\begin{cases}R(z), & z\in\mathbb{C}^+,\\
    R^+(z),& z\in U,\end{cases}
  $$
  is locally H\"{o}lder continuous in the same topology\/ \textup{(}with exponent
  $\alpha$\textup{)}.

  A similar statement applies when $\mathbb{C}^+$ is replaced by
  $\mathbb{C}^-$, but note that the limiting values $R^\pm(\lambda)$ are in
  general different.
\end{thm}
    \begin{rem}  In view of the Stieltjes formula we have
  \be
  \label{A=Rpm}A(\lambda)=\frac{1}{2\pi i}\left(R^+(\lambda)-
    R^-(\lambda)\right),\quad \lambda\in U.
  \ee
  In particular, $H$ is absolutely continuous in $U.$
   \end{rem}
    We now consider a perturbation by a potential function $V(x).$ To deal with the requirements on $V$ in this framework we introduce the following definition.
\begin{defn}\label{defn-Vsr}
  An operator $V:\MX^*_H\to\MX$ will be called
  \begin{enumerate}
  \item \textit{Short-range with respect to $H$\/} if it is compact.
  \item \textit{Symmetric\/} if $D(H)\cap\MX^*_H$ is dense in $\Lcal^2(\Rn,\CC^K)$ and the
    restriction of $V$ to $D(H)\cap\MX^*_H$ is symmetric in $\Lcal^2(\Rn,\CC^K).$
  \end{enumerate}
\end{defn}

The following lemma shows that (with some additional assumption) the operator
$H+V$ is well defined.

\begin{lem}\label{lem-sadj}
  Let $H$ be of type $\left(\MX,\MX^*_H,\alpha, U\right)$ and let $V$ be short-range and
  symmetric. Suppose that there exists $z\in\mathbb{C}$, $\Im z\neq 0$, and a
  linear subspace $D_z\subseteq D(H)\cap\MX^*_H$ such that $(H-z)(D_z)$, the
  image of $D_z$ under $H-z$, is dense in $\MX$.

  Then $P=H+V$, defined on $D(H)\cap\MX^*_H$, is essentially self-adjoint.
\end{lem}
In what follows we always assume that $H$ is of type $\left(\MX,\MX^*_H,
  \alpha, U\right)$ and that $V$ is short-range and symmetric. Thus, by the
lemma, $P=H+V$ can be extended as a self-adjoint operator with domain
$D(P)\supseteq D(H)\cap\MX^*_H$, and we retain the notation $P$ for this
extension.

Our aim is to study the spectral properties of $P$, particularly the LAP, in
this abstract framework.

Denote by $S(z)=(P-z)^{-1}$, $\Im z\neq 0$, the resolvent of $P$.
Our starting point is the resolvent equation
\[
S(z)(I+VR(z))=R(z),\quad \Im z\neq 0.
\]
It can be shown that the inverse $(I+VR(z))^{-1}$
exists on $\MX$ if $\Im z\neq 0$. This leads to
\be\label{Sz}
S(z)=R(z)(I+VR(z))^{-1}, \ee
where the equality is certainly valid from
$\MX\to\MX^*_H$.

Suppose now that $\lambda\in U$. In view of Theorem ~\ref{th-LAP-ABS} and the
assumption on $V$ we have
\[
\lim_{\eps\downarrow 0}VR(\lambda \pm i\eps)=VR^\pm(\lambda) \enspace \text{in
  $B(\MX)$.}
\]
Thus, if $(I+VR^\pm(\lambda))^{-1}$ exists (in $B(\MX)$), then Eq.
~\eqref{Sz} implies the existence of the limits \be
\label{spm}
S^\pm(\lambda)=\lim_{\eps\downarrow 0}S(\lambda \pm
i\eps)=R^\pm(\lambda)(I+VR^\pm(\lambda))^{-1},
\ee
in the uniform operator topology of $B(\MX,\MX^*_H)$.

Let $\lambda\in U$ be a point at which, say, $(I+VR^+(\lambda))^{-1}$ does not
exist (in $B(\MX)$). Since $VR^+(\lambda)$ is compact in $\MX$, there exists a
non-zero $\phi\in\MX$ so that
$$\phi=-VR^+(\lambda)\phi.$$
Let $\psi=R^+(\lambda)\phi\in\MX^*_H$. Then
$$<\psi,\phi>\,=-\lim_{\eps\downarrow 0}\left(R(\lambda+
i\eps)\phi,VR(\lambda+ i\eps)\phi\right).
$$
By the symmetry of $V$ the right-hand side of this equality is real, so we
conclude that $\Im<R^+(\lambda)\phi,\phi>\,=0$, and invoking
Eq.~\eqref{A=Rpm} we conclude that
\[
<A(\lambda)\phi,\phi>\,=0.
\]
Now the form $<A(\lambda)x,y>\,=\frac{d}{d\lambda}\left(E(\lambda)x,y\right)$ on
$\MX\times\MX$ is symmetric and positive semi-definite.  Hence, for every
$y\in\MX$,
\be
\label{Apos}
\left|<A(\lambda)\phi,y>\right|\leq\;<A(\lambda)\phi,\phi>^\frac12 \,
<A(\lambda)y,y>^\frac12=0,
\ee
and we conclude that
\be
\label{Aphiphi}
A(\lambda)\phi=0.
\ee
In particular, $R^+(\lambda)\phi=R^-(\lambda)\phi$ and
$$\phi=-VR^\pm(\lambda)\phi.$$
\begin{defn}\label{sigmap}
  We designate by $\Sigma_P$ the set
  $$\Sigma_P=\bigl\{\lambda\in
    U\,\,/\,\,\text{There exists a non-zero $\phi_\lambda\in\MX$ such
      that $\phi_\lambda=-VR^\pm(\lambda)\phi_\lambda$}\bigr\}.$$
\end{defn}

\begin{rem}
  The set $\Sigma_P$ is (relatively) closed in $U$.  Indeed, if
  $(I+VR^\pm(\lambda_0))^{-1}$ exists (in $B(\MX)$), then
  $(I+VR^\pm(\lambda))^{-1}$ exists for $\lambda$ close to $\lambda_0$.
\end{rem}

The discussion above leads to the following theorem.

\begin{thm}\label{thm-LAP-perturb}
  The operator $P=H+V$ satisfies the\/ \textup{LAP} in $U\setminus \Sigma_P$, in the
  uniform operator topology of $B(\MX,\MX^*_H)$, and the limiting values of
  its resolvent there are given by\/ \textup{Eq. ~\eqref{spm}}.
\end{thm}

In particular, the spectrum of $P$ in $U\setminus \Sigma_P$ is absolutely
continuous. We single out this fact, stated in terms of the eigenvalues, in
the following corollary.

\begin{cor}\label{specp}
  Let $\sigma_p(P)$ be the point spectrum of $P$. Then
  $$\sigma_p(P)\cap U\subseteq
  \Sigma_P.$$
\end{cor}

\subsubsection{\textbf{The exceptional set \boldmath
  $\Sigma_P$}}\label{secexceptionalset}

Our aim is to identify the set $\Sigma_P$ introduced in Definition
~\ref{sigmap}. It will turn out that (modulo one additional assumption on the
smoothness of the spectral measure of $H$) we have equality of the sets in the
last corollary. In other words, $\Sigma_P$ is the set of eigenvalues of $P$
embedded in $U$, and is necessarily discrete.

Let $\mu\in \Sigma_P$, so that
by definition there exists a non-zero $\phi\in\MX$ satisfying
\be\label{muphi}\phi=-VR^\pm(\mu)\phi.\ee In view of ~\eqref{Aphiphi} we have
$A(\mu)\phi=0$, and since the form $<A(\lambda)\phi,\phi>$ is non-negative we
infer that the zero at $\lambda=\mu$ is a minimum. Thus formally this minimum
is a second-order zero for the form. However, our smoothness assumption on the
spectral measure (Definition ~\ref{defn-typeU}) does not go so far as a
second-order derivative.  We therefore impose the following additional
hypothesis on the spectral derivative near such a minimum.

\begin{THEOREM S}
  Let $K\subseteq U$ be compact and $\phi\in\MX$ a solution to ~\eqref{muphi},
  where $\mu\in K$. Then there exist constants $C,\,\eps>0$, depending only on
  $K$, so that \be \label{asss}<A(\lambda)\phi,\phi>\,\leq\,
  C\left|\lambda-\mu\right|^{1+\eps}\left\|\phi\right\|_{\MX}^2,\quad \lambda\in K.  \ee
\end{THEOREM S}
\begin{rem}
   This assumption  is satisfied if the operator-valued function $\lambda\to
  A(\lambda)\in B(\MX,\MX_H^*)$ has a H\"{o}lder continuous (in the uniform
  operator topology) Fr\'{e}chet derivative in a neighborhood of $\mu$.
  Indeed, in this case we have
  $$ <A(\lambda)\phi,\phi>\,=\,
  <(A(\lambda)-A(\mu))\phi,\phi>\,=\,\frac{d}{d\theta}
  <A(\theta)\phi,\phi>_{\theta\in[\mu,\lambda]}\left(\lambda-\mu\right).$$
\end{rem}

\begin{thm}\label{thm-sigmap}Let $V$ be symmetric and short-range, and
  assume that the condition of\/ \textup{Lemma ~\ref{lem-sadj}} is satisfied, so
  that $P=H+V$ is a self-adjoint operator. Assume, in addition, that the assumption above is satisfied and ~\eqref{asss}
   holds. Then
  $$\Sigma_P=\sigma_p(P)\cap U.$$
   Furthermore, every eigenvalue is of finite multiplicity and the set of eigenvalues $\sigma_p(P)$ has no accumulation point in $U.$
\end{thm}

              \subsection{\textbf{PERTURBATION OF THE DIRAC OPERATOR}}\label{subsection-spec-dirac-perturb}.\newline
                We now consider the operator
                 \be
                    H_m^V=H_m+V(x),
                 \ee
                   where $H_m$ is the free Dirac operator given in ~\eqref{diracop} and  $V(x)$ is a Hermitian $4\times 4$ matrix  that decays as $|x|\to\infty.$

                   We assume $m>0.$ In this case  the spectrum $spec(H_m)$ has a ``gap'' $(-m,m)$ (see Equation ~\eqref{eqsepecHm}) and our focus is on the \textit{finiteness} of the eigenvalues therein. We refer to ~\cite{berthiert,birman-laptev,dolbo2} for general discussion of the eigenvalues in the gap.  A variational characterization of the eigenvalues is given in ~\cite{dolbo1}.

                   The spectral structure of $H_m$ was studied in Subsection ~\ref{subsection-spec-dirac}. Recall that by
                   ~\eqref{dom Hm} the domain of $H_m$ (as a self adjoint operator in $\Lcal^2(\RT,\mathbb{C}^4)$) is
                   $\Hcal^1(\RT,\mathbb{C}^4).$ Using the terminology introduced in Subsection ~\ref{subsection-spec-general-perturb} it follows that , in view of Theorem ~\ref{basiclap}, we can take
                   \be
                   \MX=\Lcal^{2,s}(\RT,\mathbb{C}^4),\quad \MX^*=\Lcal^{2,-s}(\RT,\mathbb{C}^4),\,\,s>\frac12.
                   \ee
                   Also, the space $\MX^*_H$ can then be taken as the domain of (the closure of) $H_m$ in $\Lcal^{2,-s}(\RT,\mathbb{C}^4),$
                   so that
                   \be\label{eqMXstarDirac}
                     \MX^*_H=\Hcal^{1,-s}(\RT,\mathbb{C}^4).
                   \ee
                      For the selfadjointness of $H_m^V$ the following proposition suffices for our study here.
                      \begin{prop}
                      Assume that $V$ is bounded and decays at infinity.
                      Then the operator $H_m^V$ is self-adjoint with domain $\Hcal^1(\RT,\mathbb{C}^4)\subseteq \Lcal^2(\RT,\mathbb{C}^4).$ Its essential spectrum is $\RR\setminus(-m,m).$
                      \end{prop}
                      \begin{proof}
                       By the coercivity inequality ~\eqref{eq-dirac-coerciv} it follows that $V$ is compact with respect to $H_m,$ which establishes the claim by the general theory of self-adjoint perturbations.
                      \end{proof}
                       \begin{rem}
                    Indeed, the self-adjointness of $H_m^V$ holds for a much wider class of potentials, see e.g. ~\cite[Chapter 2]{balinsky}, ~\cite{esteban}. As in the case of the classical Scr\"{o}dinger operators , the stronger decay
                    assumptions are needed when studying the LAP and the spectral derivative associated with the operator, as we proceed to do next.
                    \end{rem}
                 In order to derive the Limiting Absorption Principle (LAP) of the perturbed operator $H_m^V$ from the general theory presented in Section ~\ref{subsection-spec-general-perturb}, we need to assume that, for some constant $C>0,$
                   \be\label{eq-dirac-V-decay}
                    |V(x)|\leq C(1+|x|)^{-(1+\eps)},\quad \eps>0,\,\,x\in \RT.
                   \ee
                    It is readily verified that $V$ satisfies the conditions of Definition ~\ref{defn-Vsr}, so that it is short-range and symmetric.

                    Now let
                    $$U=\RR\setminus[-m,m].$$
                    In view of Theorem ~\ref{basiclap} and Theorem ~\ref{thm-LAP-perturb} we obtain the LAP for
                    $H_m^V$ as follows.
                    \begin{thm}\label{thm-lap-dirac-perturb} Assume that $V$ satisfies the decay assumption ~\eqref{eq-dirac-V-decay}.
    Let $R_m^V(z)=(H_m^V-z)^{-1},\,\,\Im \,z\neq 0.$ Then:
    \begin{itemize}
     \item For any $\frac12<s<\frac{1+\eps}{2}$
    the limits
    \be\label{eqLAPDiracperturb}
    R_m^{V,\pm}(\mu)=\lim\limits_{\eps\downarrow 0}R_m^V(\mu\pm
    i\eps),\,\,\mu\in U\setminus\Sigma_{H_m}^V,
    \ee
    exist in the uniform operator topology of
    $B(\Lcal^{2,s}(\RT,\mathbb{C}^4),\Lcal^{2,-s}(\RT,\mathbb{C}^4)),$ where the (relatively) closed set $\Sigma_{H_m}^V$ is given by
      $$\Sigma_{H_m}^V=\bigl\{\lambda\in
    U\,\,/\,\,\text{There exists a non-zero $\phi_\lambda\in\Lcal^{2,s}(\RT,\mathbb{C}^4)$ such
      that $\phi_\lambda=-VR_m^\pm(\lambda)\phi_\lambda$}\bigr\}.$$
      \item  Assume that the decay assumption ~\eqref{eq-dirac-V-decay} is replaced by the stronger one
       \be\label{eq-dirac-V-decay-strong}
                    |V(x)|\leq C(1+|x|)^{-(2+\eps)},\quad \eps>0,\,\,x\in \RT,
                   \ee
                   then
                    for any $1<s<\frac{2+\eps}{2}$  the limits in ~\eqref{eqLAPDiracperturb} exist for all $\mu\in\RR\setminus\widetilde{\Sigma_{H_m}^V},$
                     where now
                    $$\widetilde{\Sigma_{H_m}^V}=\bigl\{\lambda\in
    \RR\,\,/\,\,\text{There exists a non-zero $\phi_\lambda\in\Lcal^{2,s}(\RT,\mathbb{C}^4)$ such
      that $\phi_\lambda=-VR_m^\pm(\lambda)\phi_\lambda$}\bigr\}.$$
      \end{itemize}
     \end{thm}
        \begin{proof}
        The existence of the limits in   ~\eqref{eqLAPDiracperturb} in $U\setminus\Sigma_{H_m}^V$ (resp. in
                   $\RR\setminus\widetilde{\Sigma_{H_m}^V}$) for the decay rate ~\eqref{eq-dirac-V-decay} (resp. ~\eqref{eq-dirac-V-decay-strong}) follows from the general theory (see Theorem ~\ref{thm-LAP-perturb}) and the estimates implied by Theorem ~\ref{basiclap}.
                   \end{proof}
                   \begin{rem} Note that the operator $H_m^V$ is certainly defined on a dense subspace of $\Lcal^{2,-s}(\RT,\mathbb{C}^4).$ If it is closable and the graph-norm of its closure is equivalent to the norm $\Hcal^{1,-s}(\RT,\mathbb{C}^4),$ then the general theory (as in the case of the free operator, see Equation ~\eqref{eqMXstarDirac}) implies that the limits ~\eqref{eqLAPDiracperturb} are actually obtained in the  uniform operator topology of
    $B(\Lcal^{2,s}(\RT,\mathbb{C}^4),\Hcal^{1,-s}(\RT,\mathbb{C}^4)).$
                   \end{rem}

%
%
%

                     \subsubsection[The discreteness and finiteness of eigenvalues]{The discreteness and finiteness of eigenvalues}\label{secfiniteevs}

                   Our next goal is to identify the set $\Sigma_{H_m}^V$ with the (discrete) set of embedded eigenvalues. The abstract framework that enables us to do that was introduced in Subsubsection ~\ref{secexceptionalset}. We shall need to replace the short-range condition ~\eqref{eq-dirac-V-decay} by  the stronger decay condition ~\eqref{eq-dirac-V-decay-strong} on the potential $V,$
                    even when dealing with spectral intervals that do not include the thresholds at $\pm m.$
    \begin{thm}\label{thmlapVinterior}     Let   $$U=\RR\setminus[-m,m].$$  Suppose that the potential $V$ satisfies  the following decay condition:

                   \be\label{eq-dirac-V-decay-stronger1}
                    |V(x)|\leq C(1+|x|)^{-(2+\eps)},\quad \eps>0,\,\,x\in \RT.
                   \ee

                   Let $\Sigma_{H_m}^V$ be as in Theorem ~\ref{thm-lap-dirac-perturb}. Then
                   \be\label{eqdiracsigmaevs} \Sigma_{H_m}^V=\sigma_p(H_m^V)=\text{the set of eigenvalues in}\,\,U\,,\ee
                   and this set is discrete, with all eigenvalues of finite multiplicity.
                   \end{thm}
                   \begin{proof}

                    In view of Theorem ~\ref{thm-sigmap}, we need to verify that the regularity condition  ~\eqref{asss} is satisfied under the decay assumption ~\eqref{eq-dirac-V-decay-stronger1} .

                   The spectral derivative is given by Equation ~\eqref{A0} (assuming $\lambda>m$):
                   \be\label{eqalammu}
     <A_m(\lambda)f,f>=\frac{\lambda}{\sqrt{\lambda^2-m^2}}
     \int\limits_{|\xi|=\sqrt{\lambda^2-m^2}}
|(\mathcal{G}_mf)_+(\xi)|^2d\Sigma_{\sqrt{\lambda^2-m^2}}.
    \ee

     Let $\frac12<s<\frac{1+\eps}{2}.$

         Suppose that $\phi_\mu\in\Lcal^{2,s}(\RT,\mathbb{C}^4)$ is such
      that $\phi_\mu=-VR_m^\pm(\mu)\phi_\mu, \, \mu>m.$

      Thus $<A_m(\mu)\phi_\mu,\phi_\mu>=0,$
            so that $(\mathcal{G}_m\phi_\mu)_+(\xi)=0$ for all $\xi$ on the sphere $|\xi|=\sqrt{\mu^2-m^2}.$
            We need to show that
            \be\label{eqAmphimuholder}
            <A_m(\lambda)\phi_\mu,\phi_\mu>\quad\leq C|\lambda-\mu|^{1+\eps}.
            \ee
            We note that by Theorem  ~\ref{basiclap} we have $R_m^\pm(\mu)\phi_\mu\in \Hcal^{1,-s}(\RT,\mathbb{C}^4).$
             Therefore ~\eqref{eq-dirac-V-decay-stronger1} implies that
            $\phi_\mu=-VR_m^\pm(\mu)\phi_\mu\in \Lcal^{2,-s+2+\eps}(\RT,\mathbb{C}^4).$ It follows by ~\eqref{eqdefGm}
            that $\mathcal{G}_m\phi_\mu\in \Hcal^\theta,$ where $\theta>2+\eps-\frac{1+\eps}{2}=\frac32+\frac12\eps.$

            In particular, the trace of $\nabla_\xi\mathcal{G}_m\phi_\mu(\xi)$ on the sphere $|\xi|=\sqrt{\lambda^2-m^2}$
            is H\"{o}lder continuous (for $\lambda$ near $\mu$). From Equation ~\eqref{eqalammu} we infer that
            $\frac{d}{d\lambda}<A_m(\lambda)\phi_\mu,\phi_\mu>$ is  H\"{o}lder continuous in a neighborhood of $\mu.$
             Since the nonnegative function $<A_m(\lambda)\phi_\mu,\phi_\mu>$ has a minimum at $\lambda=\mu,$ its derivative vanishes there, which yields ~\eqref{eqAmphimuholder} in view of the  H\"{o}lder continuity.

                   \end{proof}
                   Finally we treat the case of the full line (and in particular the thresholds at $\lambda=\pm m$ are included). This is related to the second part of Theorem
                   ~\ref{thm-lap-dirac-perturb}. In fact,
                   the following theorem can be viewed as the ``perturbed'' version of the second part of
                   Theorem  ~\ref{basiclap}.
                   \begin{thm}Suppose that the potential $V$ satisfies  the following decay condition:

                   \be\label{eq-dirac-V-decay-stronger2}
                    |V(x)|\leq C(1+|x|)^{-(3+\eps)},\quad \eps>0,\,\,x\in \RT.
                   \ee
                   Let $\widetilde{\Sigma_{H_m}^V}$ be as in Theorem ~\ref{thm-lap-dirac-perturb}. Then
                   \be\label{eqdiracsigmaevs1} \widetilde{\Sigma_{H_m}^V}=\sigma_p(H_m^V)=\text{the set of eigenvalues in}\,\,\RR\,,\ee
                   and this set is discrete, with all eigenvalues of finite multiplicity.

                   \end{thm}
                   \begin{proof}
                        We follow the line of proof of Theorem ~\ref{thmlapVinterior} , subject to modifications needed due to the fact that we are now looking at a full neighborhood of the thresholds.

                        Without loss of generality let us consider a neighborhood of $\lambda=m.$

                         Let $1<s<\frac{2+\eps}{2}.$

                         According to the second part of Theorem
                   ~\ref{thm-lap-dirac-perturb} the limits ~\eqref{eqLAPDiracperturb} exist in $\RR\setminus\Sigma_{H_m}^V.$

                   Suppose that $\phi_\mu\in\Lcal^{2,s}(\RT,\mathbb{C}^4)$ is such
      that $\phi_\mu=-VR_m^\pm(\mu)\phi_\mu, \, $ where $\mu\in(m-\eta,m+\eta),$ for some small $\eta>0.$

      Thus $<A_m(\mu)\phi_\mu,\phi_\mu>=0.$
            For $\mu\geq m$ this means  that $(\mathcal{G}_m\phi_\mu)_+(\xi)=0$ for all $\xi$ on the sphere $|\xi|=\sqrt{\mu^2-m^2}.$

            In view of Theorem ~\ref{thm-sigmap}, we need to verify that the regularity condition  ~\eqref{asss} is satisfied under the decay assumption ~\eqref{eq-dirac-V-decay-stronger2} .

            We need to show that, for some $C>0$ depending only on $\eta,$
            \be\label{eqAmphimuholder1}
            <A_m(\lambda)\phi_\mu,\phi_\mu>\quad\leq C|\lambda-\mu|^{1+\eps},\quad \lambda,\,\mu\in (m-\eta,m+\eta).
            \ee
              Remark that $<A_m(\lambda)\phi_\mu,\phi\,\,_\mu>=0$ for $\lambda\in [-m,m].$

            We note that by Theorem  ~\ref{basiclap} we have $R_m^\pm(\mu)\phi_\mu\in \Hcal^{1,-s}(\RT,\mathbb{C}^4).$
             Therefore ~\eqref{eq-dirac-V-decay-stronger2} implies that
            $\phi_\mu=-VR_m^\pm(\mu)\phi_\mu\in \Lcal^{2,-s+3+\eps}(\RT,\mathbb{C}^4).$ It follows by ~\eqref{eqdefGm}
            that $\mathcal{G}_m\phi_\mu\in \Hcal^\theta,$ where $\theta>3+\eps-\frac{2+\eps}{2}=2+\frac12\eps.$

            In particular, the trace of $\nabla_\xi\mathcal{G}_m\phi_\mu(\xi)$ on the sphere $|\xi|=\sqrt{\lambda^2-m^2}$
            can be estimated (for $\lambda$ near $\mu$). In view of the trace Lemma \ref{tracelem} and Equation ~\eqref{eqalammu} we infer that, for $\lambda\in[m,m+\eta),$ the derivative $\frac{d}{d\lambda}<A_m(\lambda)\phi_\mu,\phi_\mu>$ is H\"{o}lder continuous and
            \be |\frac{d}{d\lambda}<A_m(\lambda)\phi_\mu,\phi_\mu>|\leq C\Big(\sqrt{\lambda^2-m^2}\Big)^{2+\eps-1-1}=C(\lambda^2-m^2)^{\frac{\eps}{2}}.\ee

             Now $$<A_m(\mu)\phi_\mu,\phi_\mu>=\frac{d}{d\lambda}<A_m(\lambda)\phi_\mu,\phi_\mu>\Big|_{\lambda=\mu}=0,$$
             so that ~\eqref{eqAmphimuholder1} follows readily , as in the conclusion of the proof of Theorem ~\ref{thmlapVinterior}.

                   \end{proof}
                   \begin{cor}\be
                    H_m^V=H_m+V(x),
                 \ee
                   where $H_m$ is the free Dirac operator given in ~\eqref{diracop} and the potential $V(x)$ is a Hermitian $4\times 4$ matrix  satisfying the decay condition ~\eqref{eq-dirac-V-decay-stronger2}.

                   Then $H_m^V$ has at most a discrete sequence of eigenvalues in $\RR.$ In particular, it has at most a finite number of eigenvalues in the ``gap'' $(-m,m).$

                   \end{cor}
%
%
%

                   \begin{rem}\label{remfiniteevs}
                   Using a stronger decay hypothesis (typically $|V(x)|\leq C(1+|x|)^{-5-\eps}$) it was shown in ~\cite[Proposition 2.3]{boussaid} that there are no eigenvalues of $H_m^V$ near the thresholds at $\pm m.$

                    More generally, Cojuhari ~\cite{cojuhari} established the finiteness of the number of eigenvalues in the gap under the weaker condition ~\eqref{eq-dirac-V-decay-stronger1} (which in our study served only for the discreteness of \textit{embedded} eigenvalues).

                   \end{rem}
                   \begin{rem} For a class of potentials, including the Coulomb potential, it was shown ~\cite[Theorem 4.21]{thaller} that there are actually no embedded eigenvalues in the essential spectrum $\RR\setminus(-m,m).$

                   On the other hand, for the special potential $V(x)=\gamma/|x|,\,|\gamma|\leq \frac{\sqrt{3}}{2},$ it is known (~\cite[Section 3.1]{balinsky},~\cite{dolbo2}) that there is an infinite sequence of eigenvalues in the gap $(-m,m).$

                   \end{rem}

                \section{\textbf{GLOBAL SPACETIME ESTIMATES --FREE DIRAC and MAXWELL OPERATORS}}\label{secspacetime}
                As in the case of the spectral study above, we
                deal first with spacetime estimates for the free (``unperturbed'') Dirac and Maxwell
                operators. Our treatment here follows the method used for the case of generalized wave equations in ~\cite{BA1,BA2}.

                \subsection{\textbf{SPACETIME ESTIMATES FOR THE FREE DIRAC OPERATOR }}\label{subsection-free-dirac-spacetime}.\newline

               The free Dirac operator was defined in ~\eqref{diracop}.  We  consider the unitary group associated with it
                \be\label{klg}
              iu_t=H_mu,\quad t\in\RR,
                \ee
                subject to the initial condition
                \be\label{initial}
                u(x,0)=u_0(x),\quad x\in\RT,
                \ee
                where $u,\,u_0$ are $\mathbb{C}^4-$valued
                functions.

                 We refer to ~\cite{ancona2} for Strichartz estimates for the Dirac operator with magnetic potentials. Here we formulate a global spacetime estimate in the weighted$-L^2$ framework.
                \begin{thm}\label{thmspacetimedirac}  Let $m > 0 ,\,s>1, $ and let $u (x, t)$ be the
solution to ~\eqref{klg}-\eqref{initial}. There exists a constant
$C=C_{s,m}>0,$
 such that
$$  \int\limits_\RR\int\limits_{\RT} (1 + | x |^2)^{-s} | u (x, t) |^2 d x d t \leq \\
 C  \| u_0 \|^2_0.
$$
\end{thm}

\begin{proof}

The solution $u (x, t)$ can be represented as
$$ u =\exp(-itH_m)u_0.  $$
 In the proof we find it clearer to use the separate notations  $( \cdot, \cdot), [ \cdot, \cdot ]$
 for the scalar
products in $\Lcal^2 (\RT,\mathbb{C}^4), \Lcal^2
(\RR^{4},\mathbb{C}^4)$, respectively.

We denote by $$\tilde v (x, \tau) = (2 \pi)^{- \frac12}
\int\limits_\RR v (x, t) e^{- i t \tau} d t$$ the partial Fourier
transform with respect to $t.$

The spectral derivative $A_m$ was defined in equation ~\eqref{Aestimate1}.

 To estimate $u(x,t) $  we use
a duality argument. Take $w (x,t) \in C^\infty_0
(\RR^4,\mathbb{C}^4)$.  Then,
$$ \aligned &[ u, w ] = \int\limits^\infty_{- \infty}
dt \int\limits_{\RT} <(e^{-itH_m} u_0)(x,t),
w(x, t)>_{\CC^4} d x  \\
&= \int\limits^\infty_{-\infty} < A_m(\lambda)u_0,
\int\limits^\infty_{- \infty} e^{- i t \lambda} w (\cdot, t) dt > d \lambda \\
&= (2 \pi)^{1/2} \int\limits^\infty_{-\infty} < A_m(\lambda) u_0,
\tilde w (\cdot, \lambda) > d \lambda. \endaligned , $$
where in the last two terms $<,>$ is the $(\Lcal^{2,-s}(\RT,\CC^4),\Lcal^{2,s}(\RT,\CC^4))$ pairing.

    We  note that by the spectral theorem
\be\int\limits^\infty_{-\infty}  <A_m (\lambda) f, f > d
           \lambda = \| f \|_0^2.\ee

 The
Cauchy-Schwarz inequality yields

$$ \aligned \big\vert [u,w] \big\vert \hspace{250pt} \\\leq ( 2 \pi)^{1/2}
\left( \int\limits^\infty_{-\infty} < A_m(\lambda) u_0 , u_0 > d
\lambda \right)^{1/2} \cdot \left( \int\limits^\infty_{-\infty}
< A_m(\lambda) \tilde w (\cdot, \lambda), \tilde w (\cdot, \lambda) > d \lambda \right)^{1/2}  \\
= (2 \pi)^{1/2} \| u_0\|_2 \cdot \left(
\int\limits^\infty_{-\infty} < A_m (\lambda) \tilde w (\cdot,
\lambda), \tilde w (\cdot, \lambda)
> d \lambda \right)^{1/2}. \endaligned $$
and recalling
   ~\eqref{Aestimate1}
 we obtain
$$ \aligned \big\vert [u,w] \big\vert \hspace{250pt}
 \\
\leq (2 \pi)^{1/2} C\| u_0\|_0\cdot \Big(\int\limits_{|\lambda|>m}
\min(\frac{|\lambda|}{\sqrt{\lambda^2-m^2}},|\lambda|(\lambda^2-m^2)^{s-1})
\cdot \| \tilde w (\cdot, \lambda) \|^2_{0,s} d \lambda\Big)^{\frac12}.
  \endaligned $$
  By the Plancherel theorem, using $s>1,$
$$ \big\vert [u,w] \big\vert \leq (2 \pi)^{1/2}
C  \| u_0\|_0 \left( \int\limits_\RR \| w (\cdot, t ) \|^2_{0,s}
dt \right)^{1/2}.  $$

 Let $f (x,t) \in C^\infty_0
(\RR^4,\mathbb{C}^4),$ and take  $w (x,t) = ( 1 + |x|^2)^{-
\frac{s}{2}} f (x,t)$, so that
$$\int\limits_\RR \| w (\cdot, t )
\|^2_{0,s} dt = \| f \|^2_0.$$ We infer that
$$ \big\vert [(1 + | x|^2)^{- \frac{s}{2}} u , f ]
\big\vert \leq (2 \pi)^{1/2} C \cdot  \| u_0 \|_0 \cdot \| f \|_0,
$$
which concludes the proof of the theorem.
\end{proof}
\begin{rem} Note that we need only (for the unperturbed case)
                $s>1,$ including the propagation near the threshold. This is
                to be compared with Theorem 1.1 and
                Proposition 2.1 in ~\cite{boussaid}, where
                $s>\frac52,$ is assumed (including a potential).
                \end{rem}

                 \subsection{\textbf{SPACETIME ESTIMATES FOR THE FREE MAXWELL OPERATOR }}\label{subsection-free-maxwell-spacetime}.\newline

                 The free Maxwell operator $L_{maxwell}$ was introduced in  ~\eqref{eq:maxwell0}.

                  We  consider the unitary group associated with the operator,
                \be\label{eqmaxwellspacetime}
              iu_t=L_{maxwell}u,\quad t\in\RR,
                \ee
                subject to the initial condition
                \be\label{initialmaxwell}
                u(x,0)=u_0(x),\quad x\in\RT,
                \ee
                where $u,\,u_0$ are $\mathbb{C}^6-$valued
                functions.

                We further assume that
                \be\label{equoperpkernel} u_0\in (I-\mathcal{P})\Lcal^2(\RT,\mathbb{C}^6),\ee
                  where $\mathcal{P}$ is the orthogonal projection on $ker(L_{maxwell})$ (see Equation ~\eqref{dom Lmaxwell}).

                  In view of Remark ~\ref{remTETM} the orthogonality condition means that if  $u_0(x)=\binom{E_0(x)}{B_0(x)}$  (with values in $\mathbb{C}^6$) then  Fourier transform $ \binom{\widehat{E_0}(\xi)}{\widehat{B_0}(\xi)}\in span \{\Upsilon_{\pm }\},$ for every $0 \neq \xi\in\RT.$ In particular
          $$<\widehat{E_0}(\xi),\xi>_{\CC^3}=0,\quad <\widehat{B_0}(\xi),\xi>_{\CC^3}=0, \quad \xi\in\RT\setminus\set{0},$$
          so that
           $E_0$ (resp. $B_0$) is a ``TE-mode'' (resp. ``TM-mode'').

                  \begin{thm}\label{thmspacetimemaxwell}  Let $\,s>\frac12, $ and let $u (x, t)$ be the
solution to ~\eqref{eqmaxwellspacetime}-\eqref{initialmaxwell}-\eqref{equoperpkernel}. There exists a constant
$C=C_{s}>0,$
 such that
$$  \int\limits_\RR\int\limits_{\RT} (1 + | x |^2)^{-s} | u (x, t) |^2 d x d t \leq \\
 C  \| u_0 \|^2_0.
$$
\end{thm}

\begin{proof}

The solution $u (x, t)$ can be represented as
$$ u =\exp(-itL_{maxwell})u_0.  $$
 As in the  proof of Theorem ~\ref{thmspacetimedirac}, we find it clearer to use the separate notations  $( \cdot, \cdot), [ \cdot, \cdot ]$
 for the scalar
products in $\Lcal^2 (\RT,\mathbb{C}^6), \Lcal^2
(\RR^{4},\mathbb{C}^6)$, respectively.

We denote by $$\tilde v (x, \tau) = (2 \pi)^{- \frac12}
\int\limits_\RR v (x, t) e^{- i t \tau} d t$$ the partial Fourier
transform with respect to $t.$

The spectral derivative $\widetilde{A}$ was defined in equation ~\eqref{Aestimate1max}.

 To estimate $u(x,t) $  we use
a duality argument. Take $w (x,t) \in C^\infty_0
(\RR^4,\mathbb{C}^6))$.  Then,
$$ \aligned &[ u, w ] = \int\limits^\infty_{- \infty}
dt \int\limits_{\RT} <(e^{-itL_{maxwell}} u_0)(x,t),
w(x, t)>_{\CC^6} d x  \\
&= \int\limits^\infty_{-\infty} < \widetilde{A}(\lambda)u_0,
\int\limits^\infty_{- \infty} e^{- i t \lambda} w (\cdot, t) dt > d \lambda \\
&= (2 \pi)^{1/2} \int\limits^\infty_{-\infty} <\widetilde{A}(\lambda) u_0,
\tilde w (\cdot, \lambda) > d \lambda. \endaligned  $$

    We  note that by the spectral theorem
\be\int\limits^\infty_{-\infty}  <\widetilde{A} (\lambda) f, f > d
           \lambda = \| f \|_0^2.\ee

 The
Cauchy-Schwarz inequality yields

$$ \aligned \big\vert [u,w] \big\vert \hspace{250pt} \\\leq ( 2 \pi)^{1/2}
\left( \int\limits^\infty_{-\infty} < \widetilde{A}(\lambda) u_0 , u_0 > d
\lambda \right)^{1/2} \cdot \left( \int\limits^\infty_{-\infty}
< \widetilde{A}(\lambda) \tilde w (\cdot, \lambda), \tilde w (\cdot, \lambda) > d \lambda \right)^{1/2}  \\
= (2 \pi)^{1/2} \| u_0\|_2 \cdot \left(
\int\limits^\infty_{-\infty} <\widetilde{A} (\lambda) \tilde w (\cdot,
\lambda), \tilde w (\cdot, \lambda)
> d \lambda \right)^{1/2}. \endaligned $$
and recalling
   ~\eqref{Aestimate1max}
 we obtain
$$ \aligned \big\vert [u,w] \big\vert \hspace{250pt}
 \\
\leq (2 \pi)^{1/2} C\| u_0\|_0\cdot \Big(\int\limits^\infty_{-\infty}
\min(1,|\lambda|^{2s-1})
\cdot \| \tilde w (\cdot, \lambda) \|^2_{0,s} d \lambda\Big)^\frac12.
  \endaligned $$
  By the Plancherel theorem, using $s>\frac12,$
$$ \big\vert [u,w] \big\vert \leq (2 \pi)^{1/2}
C  \| u_0\|_0 \left( \int\limits_\RR \| w (\cdot, t ) \|^2_{0,s}
dt \right)^{1/2}.  $$

 Let $f (x,t) \in C^\infty_0
(\RR^4,\mathbb{C}^6),$ and take  $w (x,t) = ( 1 + |x|^2)^{-
\frac{s}{2}} f (x,t)$, so that
$$\int\limits_\RR \| w (\cdot, t )
\|^2_{0,s} dt = \| f \|^2_0.$$ We infer that
$$ \big\vert [(1 + | x|^2)^{- \frac{s}{2}} u , f ]
\big\vert \leq (2 \pi)^{1/2} C \cdot  \| u_0 \|_0 \cdot \| f \|_0,
$$
which concludes the proof of the theorem.
\end{proof}

  \section{\textbf{GLOBAL SPACETIME ESTIMATES --STRONGLY PROPAGATIVE HOMOGENEOUS SYSTEMS}}\label{secspacetimestrongprop}

 We now address the spacetime decay estimates for \textit{homogeneous} strongly propagative operators $L_0=\suml_{j=1}^nM_j^0D_j$ (see ~\eqref{systemL0}), for which the LAP was stated in Theorem ~\ref{basiclapstrongprop}. Recall that both the Dirac (zero mass) and Maxwell systems belong to this class, but the singular set $Z$ (defined in ~\eqref{eqdefZ}) is empty in both cases.

                  We  consider the associated unitary group
                \be\label{eqstpropspacetime}
              iu_t=L_{0,hom}u,\quad t\in\RR,
                \ee
                subject to the initial condition
                \be\label{initialstprop}
                u(x,0)=u_0(x),\quad x\in\Rn,
                \ee
                where $u,\,u_0$ are $\mathbb{C}^K-$valued
                functions.

                We further assume that the initial function is orthogonal to the ``stationary waves'' of the system, namely,
                \be\label{equoperpkernelstprop} u_0\in (I-\mathcal{P})\Lcal^2(\Rn,\mathbb{C}^K),\ee
                  where $\mathcal{P}$ is the orthogonal projection on $ker(L_{0,hom}).$

                   Clearly the solution $u(x,t)=e^{-itL_{0,hom}}u_0$ satisfies $u(\cdot,t)\in(I-\mathcal{P})\Lcal^2(\Rn,\mathbb{C}^K),\,t\in\RR.$

                   It is represented by the Fourier integral operator
                   \be\label{Represente-itL_0}
                   u(x,t)=(2\pi)^{-\frac{n}{2}}\int\limits_{\Rn}e^{-it\suml_{j=1}^n M^0_j\xi_j}e^{-i<\xi, x>}\widehat{u_0}(\xi)d\xi.
                   \ee
                   However, this explicit expression does not easily lend itself to asymptotic analysis by classical methods of geometric optics; the non-commutativity of the matrices $\set{M^0_j}_{j=1}^n$ requires a very detailed study of the algebraic structure of the eigenvalue manifolds and their intersections. Hence various restrictive hypotheses need to be imposed, even in the case of constant coefficients.

                   There is extensive literature concerning the (large-time) asymptotic behavior of solutions of such systems and their perturbations, using primarily the geometric optics approach. It is beyond the scope of the present paper to present an exhaustive account of these works, and we refer to ~\cite{rauch} and references therein.  In particular it is assumed there that the system is uniformly propagative ~\cite[Assumption (1.3)]{rauch}.   The rate of decay in time of the $L^2$ norm of scattered solutions in balls is shown  ~\cite[Theorem 1]{rauch} to be $O((\log t)^{-1}).$ Confining to the constant coefficient case, the global estimate ~\eqref{eqspacetimeestunifprop} obtained below yields a faster rate of decay (in integral sense).

       The presence of the singular set $Z$ makes our spacetime estimate somewhat more delicate. Recall Definition
       ~\ref{defnUpsilon} of the closed subspace $\Upsilon^s_Z\subseteq \Lcal^{2,s}(\Rn,\mathbb{C}^K),$ and note also Remark ~\ref{rembigZ}.
       We now define the closed subspace $\Upsilon_Z\subseteq \Lcal^2(\Rn,\CC^K)$ as the inverse Fourier transform of $(I-\Delta)^{\frac{s}{2}}\widehat{\Upsilon^s_Z},$ namely
  \be\Upsilon_Z=(1+|x|^2)^{\frac{s}{2}}\Upsilon^s_Z=\set{g\,/ \,g(x)=(1+|x|^2)^{\frac{s}{2}}h(x)\,\,\mbox{for some}\,\,h\in \Upsilon^s_Z}.\ee
     Let $\mathcal{E}:\Lcal^2(\Rn,\CC^K)\hookrightarrow \Upsilon_Z$ be the orthogonal projection unto this subspace.

%
%


                  \begin{thm}\label{thmspacetimeeststrprop} Assume that $L_{0,hom}$ is strongly propagative.  Let $\,s>\frac12, $ and let $u (x, t)$ be the
solution to ~\eqref{eqstpropspacetime}-\eqref{initialstprop}-\eqref{equoperpkernelstprop}.

   Define $u_Z(x,t),\,t\in\RR,$ by
   $$u_Z(x,t)=\mathcal{E}[(1+|x|^2)^{-\frac{s}{2}}  u (x, t)].$$
 Then there exists a constant
$C=C_{s,n}>0,$
 such that
\be\label{eqspacetimeeststrprop}  \int\limits_{\RR}\int\limits_{\Rn}  | u_Z (x, t) |^2 dx  d t \leq \\
 C  \| u_0 \|^2_0.
\ee
\end{thm}
\begin{rem} Note that $u_Z(x,t)\in \Upsilon_Z$ implies $(1+|x|^2)^{-\frac{s}{2}}u_Z(x,t)\in \Upsilon^s_Z$  for every $t\in\RR.$ Fix $t=t_0\in\RR.$ In light of Definition  ~\ref{defnUpsilon} , given $\eta>0$ there exists $\psi_\eta\in\Lcal^{2,s}(\Rn,\CC^K)$ so that $\widehat{\psi_\eta}\in  C^\infty_0(\Rn\setminus\Zbar,\CC^K)$ and
  $$\int\limits_{\Rn} (1+|x|^2)^s|(1+|x|^2)^{-\frac{s}{2}}u_Z(x,t_0)-\psi_\eta(x)|^2dx\leq\eta^2.$$
  Thus
  $$\int\limits_{\Rn} |u_Z(x,t_0)-(1+|x|^2)^{\frac{s}{2}}\psi_\eta(x)|^2dx\leq\eta^2,$$
  and by the Plancherel theorem
   $$\int\limits_{\Rn} |\widehat{u_Z}(\xi,t_0)-(I-\Delta)^{\frac{s}{2}}\widehat{\psi_\eta}(\xi)|^2d\xi\leq\eta^2.$$
The function $u_Z(x,t_0)$ is therefore the part of $u(x,t_0)$ whose Fourier transform is in the closure (in $\Lcal^2(\Rn,\CC^K)$) of the range of $(I-\Delta)^{\frac{s}{2}}$ acting on smooth functions ``supported away'' from the singular set $\Zbar.$


\end{rem}

\begin{proof}[\textbf{Proof of theorem:}]
By density (in $\Lcal^2(\Rn,\CC^K)$)  we may assume
\be\label{equ0inUpsilons}\widehat{u_0}(\xi)  \in C^\infty_0(\Rn\setminus\Zbar,\CC^K).\ee

The solution $u (x, t)$ can be represented as
$$ u =\exp(-itL_{0,hom})u_0.  $$
   In particular, $\widehat{u}(\cdot,t)\in C^\infty_0(\Rn\setminus\Zbar,\CC^K)$ for all $t\in\RR.$

 As in the  proofs of Theorems ~\ref{thmspacetimedirac} and ~\ref{thmspacetimemaxwell}, we find it clearer to use the separate notations  $( \cdot, \cdot), [ \cdot, \cdot ]$
 for the scalar
products in $\Lcal^2 (\Rn,\mathbb{C}^K), \Lcal^2
(\RR^{n+1},\mathbb{C}^K)$, respectively.

We denote by $$\tilde{\gamma} (x, \tau) = (2 \pi)^{- \frac12}
\int\limits_\RR \gamma(x, t) e^{- i t \tau} d t$$ the partial Fourier
transform with respect to $t.$

The spectral derivative $\widetilde{A_{0,hom}}(\lambda)\in B(\Upsilon^s_Z,\Lcal^{2,-s}(\Rn,\CC^K)),\,\,s>\frac12,$ was defined in Equation ~\eqref{eqderivE0outZ}.

 To estimate $u(x,t) $  we use
a duality argument , as in the cases of the Dirac and Maxwell systems. Take $w (x,t) \in C^\infty
(\RR^{n+1},\mathbb{C}^K),$ so that
\be\label{eqwxtinUpsilons}\widehat{w} (\xi,t) \in C^\infty_0((\Rn\setminus\Zbar)\times \RR,\CC^K).\ee  Then,
\be\label{eqw1xtinUpsilon} \aligned &[ u, w ] = \int\limits^\infty_{- \infty}
dt \int\limits_{\RT} <(e^{-itL_{0,hom}} u_0)(x,t),
w(x, t)>_{\CC^K} d x  \\
&= \int\limits^\infty_{-\infty} < \widetilde{A_{0,hom}}(\lambda)u_0,
\int\limits^\infty_{- \infty} e^{- i t \lambda} w (\cdot, t) dt > d \lambda \\
&= (2 \pi)^{1/2} \int\limits^\infty_{-\infty} <\widetilde{A_{0,hom}}(\lambda) u_0,
\tilde w (\cdot, \lambda) > d \lambda. \endaligned  \ee

    We  note that by the spectral theorem
\be\int\limits^\infty_{-\infty}  <\widetilde{A_{0,hom}} (\lambda) u_0, u_0 > d
           \lambda = \| u_0 \|_0^2.\ee

 The
Cauchy-Schwarz inequality yields

$$ \aligned \big\vert [u,w] \big\vert \hspace{250pt} \\\leq ( 2 \pi)^{1/2}
\left( \int\limits^\infty_{-\infty} < \widetilde{A_{0,hom}}(\lambda) u_0 , u_0 > d
\lambda \right)^{1/2} \cdot \left( \int\limits^\infty_{-\infty}
< \widetilde{A_{0,hom}}(\lambda) \tilde w (\cdot, \lambda), \tilde w (\cdot, \lambda) > d \lambda \right)^{1/2}  \\
= (2 \pi)^{1/2} \| u_0\|_2 \cdot \left(
\int\limits^\infty_{-\infty} <\widetilde{A_{0,hom}} (\lambda) \tilde w (\cdot,
\lambda), \tilde w (\cdot, \lambda)
> d \lambda \right)^{1/2}. \endaligned $$
and recalling
   ~\eqref{eqestAohom} and the fact  that the norm of $\Upsilon^s_Z$ is the $\Lcal^{2,s}$ norm,
 it follows that
$$  \big\vert [u,w] \big\vert
 \leq (2 \pi)^{1/2} C\| u_0\|_0\cdot \int\limits^\infty_{-\infty}
 \| \tilde w (\cdot, \lambda) \|^2_{0,s} d \lambda.
   $$
  Invoking  the Plancherel theorem, this estimate leads to
\be\label{equwestUpsilon} \big\vert [u,w] \big\vert \leq (2 \pi)^{1/2}
C  \| u_0\|_0 \left( \int\limits_\RR \| w (\cdot, t ) \|^2_{0,s}
dt \right)^{1/2}.  \ee
  The estimate ~\eqref{equwestUpsilon} was obtained under the condition ~\eqref{eqwxtinUpsilons}, and by closure it holds for all $w(x,t)\in \Lcal^2(\RR_t, \Upsilon^s_Z).$

  Now note that $v(x,t)=(1+|x|^2)^{\frac{s}{2}}w(x,t)\in \Lcal^2(\RR_t,\Upsilon_Z).$

  The estimate ~\eqref{equwestUpsilon} can be rewritten as
  \be\label{equwestUpsilon1} \big\vert [(1+|x|^2)^{-\frac{s}{2}}u,v] \big\vert \leq (2 \pi)^{1/2}
C  \| u_0\|_0 \left( \int\limits_\RR \| v (\cdot, t ) \|^2_{0}
dt \right)^{1/2},\quad  \,\,\forall\, v(\cdot,t)\in \Lcal^2(\RR_t,\Upsilon_Z), \ee
which clearly entails ~\eqref{eqspacetimeeststrprop}.
%


\end{proof}
\begin{cor}  Assume that $L_{0,hom}$ is uniformly propagative (Definition ~\ref{defn-uniform-prop}).  Let $\,s>\frac12, $ and let $u (x, t)$ be the
solution to ~\eqref{eqstpropspacetime}-\eqref{initialstprop}-\eqref{equoperpkernelstprop}.

 Then there exists a constant
$C=C_{s,n}>0,$
 such that
\be\label{eqspacetimeestunifprop}  \int\limits_{\RR}\int\limits_{\Rn} (1+|x|^2)^{-s} | u (x, t) |^2 dx  d t \leq \\
 C  \| u_0 \|^2_0.
\ee
\end{cor}
\begin{proof} In this case by definition $Z=\emptyset$ so that $u_Z(x,t)=(1+|x|^2)^{-\frac{s}{2}}  u (x, t)$ in Theorem ~\ref{thmspacetimeeststrprop}.
\end{proof}
\begin{rem}\label{remifconjecture} Note that if Conjecture ~\ref{conjecdense} is shown to be true, then $\Upsilon_Z=\Lcal^2(\Rn, \CC^K)$ and $u_Z=(1+|x|^2)^{-\frac{s}{2}}  u (x, t)$ for any strongly propagative system.
\end{rem}
   The result in Theorem ~\ref{thmspacetimeeststrprop} involved the ``projected'' function $u_Z$ because we let the initial function $u_0$ be any function in $\Lcal^2(\Rn, \mathbb{C}^K).$ However, if we restrict the support of $u_0$ away from $\Zbar$ we can get an improved estimate , but with a constant that  depends on the support of $u_0$ as follows.
   \begin{thm}\label{thmspacetimeeststrprop1} Assume that $L_{0,hom}$ is strongly propagative.  Let $\,s>\frac12, $ and let $u (x, t)$ be the
solution to ~\eqref{eqstpropspacetime}-\eqref{initialstprop}-\eqref{equoperpkernelstprop}. Let $\mathfrak{K}\Subset \Rn\setminus \Zbar$ be a compact set
and assume further that  supp$[\widehat u_0] \subseteq \mathfrak{K}.$

Then there exists a constant $C=C_{s, n,\mathfrak{K}}>0$  such that
\be  \int\limits_{\RR}\int\limits_{\Rn} (1+|x|^2)^{-s} | u (x, t) |^2 dx  d t \leq \\
 C  \| u_0 \|^2_0.
\ee

\end{thm}
\begin{proof} We repeat the proof of Theorem ~\ref{thmspacetimeeststrprop}. However instead of  ~\eqref{eqwxtinUpsilons} we take any $\widehat{w} (\xi,t) \in C^\infty_0((\RR^{n+1},\CC^K).$ Now let $\chi(\xi)\in C^\infty_0(\Rn\setminus\Zbar)$ so that $\chi\equiv 1$ on a compact neighborhood of
$\mathfrak{K}.$ Define $\widehat{w_1}(\xi,t)=\chi(\xi)\widehat{w}(\xi,t)$ and let $w_1(x,t)$ be the inverse Fourier transform. Clearly $w_1(x,t)\in \Lcal^2(\RR_t,\Upsilon^s_Z)$ and with the notation of ~\eqref{eqw1xtinUpsilon} we have
 $$[u,w_1]=[u,w].$$
    The proof now proceeds as before and is completed by noting that with a constant $C=C_{s, n,\mathfrak{K}}>0$
      $$\|w_1(\cdot,t)\|_{0,s}\leq C\|w(\cdot,t)\|_{0,s},\quad t\in\RR.$$
\end{proof}

\end{document}